\documentclass[journal]{IEEEtran}

\usepackage{amsmath,amsfonts}
\usepackage{algorithmic}
\usepackage{algorithm}
\usepackage{array}
\usepackage[caption=false,font=normalsize,labelfont=sf,textfont=sf]{subfig}
\usepackage{textcomp}
\usepackage{stfloats}
\usepackage{url}
\usepackage{verbatim}
\usepackage{graphicx}
\usepackage{cite}

\usepackage{bm}
\usepackage{url}

\usepackage{cite,graphicx,amsmath,amssymb}
\usepackage{fancyhdr}
\usepackage{hhline}
\usepackage{graphicx,graphics}
\usepackage{array,color}
\usepackage{amsmath}
\usepackage{stfloats}
\usepackage[flushleft]{threeparttable}
\usepackage{booktabs}

\newtheorem{proposition}{Proposition}
\newtheorem{remark}{Remark}

\DeclareMathOperator*{\maximize}{maximize}

\begin{document}
\title{Integrated Sensing and Communications for Low-Altitude Economy with Deterministic Sensing and Gaussian Information Signals}

\author{Xianxin~Song,~\IEEEmembership{Member,~IEEE}, Xianghao~Yu,~\IEEEmembership{Senior Member,~IEEE}, Jie~Xu,~\IEEEmembership{Fellow,~IEEE}, and~Derrick Wing Kwan Ng,~\IEEEmembership{Fellow,~IEEE}
\thanks{Part of this paper will be presented at the IEEE International Conference on Communications (ICC) 2026 \cite{song2025detection}.
}

\thanks{Xianxin Song and Xianghao Yu are with the Department of Electrical Engineering, City University of Hong Kong, Hong Kong, China (e-mail: xianxin.song@cityu.edu.hk, alex.yu@cityu.edu.hk). Xianghao Yu is the corresponding author.}
\thanks{Jie~Xu is with the School of Science and Engineering (SSE), the Shenzhen Future Network of Intelligence Institute (FNii-Shenzhen), and the Guangdong Provincial Key Laboratory of Future Networks of Intelligence, The Chinese University of Hong Kong (Shenzhen), Guangdong  518172, China (e-mail: xujie@cuhk.edu.cn).}
\thanks{Derrick Wing Kwan Ng is with the School of Electrical Engineering and Telecommunications, University of New South Wales, Sydney, NSW 2052, Australia (e-mail: w.k.ng@unsw.edu.au).}
}

\maketitle
\begin{abstract}
Reliable surveillance and communication for unmanned aerial vehicles (UAVs) are crucial for enabling and sustaining the accelerated growth of the low-altitude economy. The widespread deployment of base station (BS) infrastructure, combined with recent advancements in integrated sensing and communications (ISAC), offers a cost-effective and scalable framework for target sensing by leveraging existing wireless communication systems. This paper investigates a bistatic downlink ISAC architecture tailored to UAV operations, in which a BS communicates with a legitimate UAV and detects a potential unauthorized intruder in the surveillance region. We assume that the BS transmits superimposed ISAC waveforms comprising both Gaussian-information-bearing and deterministic sensing components to simultaneously provide communication and sensing functionalities within a unified paradigm. First, we develop a Neyman-Pearson (NP)-based optimal detector that jointly exploits both deterministic sensing and stochastic signal components, with closed-form analytical results characterizing their respective and joint contributions to overall detection performance. Subsequently, we optimize the transmit beamforming design at the BS  to maximize the minimum detection probability over the entire surveillance region, subject to a minimum signal-to-interference-plus-noise ratio (SINR) requirement at the authorized UAV and a total transmit power budget at the BS. The resulting design problem is highly non-convex due to the coupled design variables and the intricate structure of the detection metric, which is efficiently addressed via semi-definite relaxation (SDR) and successive convex approximation (SCA) techniques. Simulation results demonstrate the superiority of the proposed NP-based detector, which fully leverages the synergy between both types of signals, over conventional benchmark schemes that treat information-bearing signals merely as interference. Furthermore, the results reveal a fundamental sensing-communication trade-off, where increasing the communication-rate threshold directs more transmit power to Gaussian-information-bearing signals, thereby reducing the power allocated to deterministic components and consequently weakening detection performance.
\end{abstract}
\IEEEpeerreviewmaketitle
\section{Introduction}
The rapid development of unmanned aerial vehicles (UAVs) has dramatically expanded the scope of human activities from ground-based environments to lower-altitude airspace. This transformation has spawned a wide range of emerging applications such as low-altitude transportation and logistics, environmental sensing, and precision agriculture. However, the sustainable development of the low-altitude economy (LAE) critically depends on reliable and real-time UAV monitoring and communication capabilities, as unauthorized UAVs can introduce substantial risks  to public safety, airspace management, and infrastructure security \cite{10955337,11098638,11131292}.

Integrated sensing and communications (ISAC) has been recognized as one of the key enabling technologies for future sixth-generation (6G) networks \cite{9696263,3GPP_ISAC,9737357,song2025overview,11358925}, facilitating the seamless amalgamation of sensing and communication functionalities within a unified framework. Given the substantial geographic overlap between cellular communication networks and the LAE zone, leveraging existing cellular infrastructure for UAV monitoring and communication offers significant advantages in terms of deployment cost efficiency, scalability, and service coverage. However, communication and sensing are governed by fundamentally different signal design principles. Specifically, sensing systems typically utilize deterministic signals to achieve  high-resolution parameter estimation and reliable target detection\cite{richards2005fundamentals,steven1993fundamentals,1703855}, whereas communication systems rely on Gaussian-information-bearing signals to maximize achievable communication rates in additive white Gaussian noise (AWGN) channels\cite{goldsmith2005wireless}. Furthermore, in target-sensing applications, echo signals reflected by targets, whether induced by deterministic or Gaussian signals, both contain exploitable target-related information. This observation suggests that jointly leveraging deterministic and Gaussian signal components can inherently  potentially enhance sensing performance. By contrast, from the perspective of communication, only Gaussian-information-bearing signals can efficiently carry information, while deterministic sensing signals not only compete for valuable resources with communication signals, but also introduce additional interference to communication receivers. Therefore, ISAC waveform design must carefully balance the non-trivial sensing-communication trade-off to achieve dual functionality. 

To exploit the advantages of the integrated ISAC paradigm for the LAE, prior works \cite{10977743,11296935,10879807,11072035,11404407} have employed terrestrial base stations (BSs) to enable both UAV surveillance and communications. Specifically, the work in \cite{10977743} investigated a cooperative sensing scenario in which a single BS transmits sensing signals known to all receivers, while multiple BSs jointly estimate the UAV’s angle, distance, and velocity by processing the received echo signals in a distributed manner. Besides, \cite{11296935} proposed a three-phases ISAC protocol, comprising downlink sensing, uplink feedback, and downlink data communication, where pilot signals with fixed amplitude are employed during the sensing phase for target location estimation. Furthermore, the authors in \cite{10879807} studied a networked ISAC system, where multiple BSs cooperatively perform downlink communications with multiple authorized UAVs while sensing a given region. In particular, the transmit beamforming at the BSs and the UAV trajectories were jointly optimized to maximize the average communication rate, subject to minimum illumination power requirements for sensing, maximum transmit power constraints at the BSs, and UAV trajectory constraints. Moreover, in \cite{11072035}, an LAE-oriented ISAC system consisting of one BS, multiple communication UAVs, and one low-altitude target was investigated, where deep reinforcement learning was adopted to jointly design BS beamforming and UAV trajectories to maximize the expected communication sum rate while satisfying a minimum average sensing signal-to-noise ratio (SNR) requirement. Finally, \cite{11404407} examined beamforming prediction for ISAC-enabled UAV systems, where the BS leverages echo signals to adaptively align its transmit beams toward UAVs, to maximize the achievable communication rate.

The aforementioned works  \cite{10977743,11296935,10879807,11072035,11404407} can be broadly classified into two categories. The first category exclusively utilizes only dedicated deterministic sensing signals for target sensing \cite{10977743,11296935}, enabling direct application of conventional sensing performance metrics. The second category employs information-bearing signals or composite ISAC signals containing both information-bearing and sensing components for sensing\cite{10879807,11072035,11404407}. These works mainly focus on monostatic sensing scenarios and typically assume sufficiently long sensing durations, thereby neglecting the impact of inherent randomness in the information-bearing signals on sensing performance. However, bridging the gap from theory to practical employment remains challenging for UAV sensing, due to the distinctive characteristics of UAVs. On the one hand, the small radar cross section (RCS) of UAVs results in weak reflected echo signals, while numerous scatterers near terrestrial BSs introduce significant clutter interference at the sensing receiver. Consequently, conventional monostatic sensing architectures may experience severely degraded detection performance. In contrast, the spatial separation between transmitter and receiver in a bistatic configuration allows the sensing receiver to be positioned in low‑clutter environments, such as elevated towers or UAVs, which is particularly advantageous for UAV sensing. On the other hand, when information-bearing signals are employed for sensing, their inherent randomness can substantially degrade detection performance if the sensing duration is not sufficiently long in monostatic systems \cite{10147248,10206462,10596930,10645253,10977963,11087656}  or when exact knowledge of the transmitted signal is unknown at the receiver in bistatic systems\cite{11391525,xie2025bistatic}. Motivated by these considerations, we focus on a bistatic sensing architecture for UAV detection and specifically analyze the sensing performance when random information-bearing signals are exploited.
 
Prior works \cite{11391525,xie2025bistatic} have investigated the ISAC performance boundaries in bistatic configurations, accounting for the sensing performance degradation caused by the randomness of information-bearing signals. In practical bistatic sensing systems, the sensing transmitter and receiver are spatially separated. In such configurations, it is commonly assumed that sensing receivers possess only the statistical properties of the random signals, rather than the exact transmitted signal realizations\cite{11391525,xie2025bistatic}. This assumption arises since forwarding exact random signal realizations to the  sensing receiver over backhaul links would incur signaling overhead, increase transmission latency, and introduce potential security vulnerabilities. Specifically, the work in \cite{11391525} focused on bistatic direction-of-arrival (DoA) estimation by exploiting superimposed deterministic and Gaussian signals. Here, the corresponding estimation Cram\'er-Rao bound (CRB)  was derived and minimized under communication-rate constraints. Furthermore, the authors in \cite{xie2025bistatic} analyzed a bistatic detection framework employing time-division multiplexing of deterministic pilots and random data payloads, and developed a generalized likelihood ratio test (GLRT) detector utilizing both pilots and data payloads. However, since the length of pilot signals is fixed and typically much shorter than that of the data payloads, the overall achievable ISAC performance remains significantly limited. Motivated by these limitations, we aim to analyze bistatic UAV detection performance by jointly leveraging deterministic sensing and Gaussian-information-bearing signals over the entire ISAC duration, allowing for more flexible adjustment of sensing performance.

This paper investigates a bistatic ISAC-assisted UAV detection and communication system. In particular, the BS transmits a superimposed waveform comprising Gaussian-information-bearing and deterministic sensing signals, while a dedicated sensing receiver processes the received echo signals to determine the presence of a potential unauthorized UAV in the surveillance region. The sensing receiver is assumed to equip with perfect knowledge of the deterministic sensing signal realizations, but only statistical knowledge of the Gaussian-information-bearing signals. Under this setup, we aim to characterize the fundamental trade-off between detection and communication performance through the design of transmit beamforming. However, this problem is particularly challenging due to the intricate integral form of the derived detection probability function, which arises from the coexistence of these two heterogeneous signal components and leads to a highly non-convex optimization structure. 

Compared with our conference version \cite{song2025detection}, this paper extends the study in two major directions. First, unlike \cite{song2025detection}, which considered target detection at a known location with both Gaussian-information-bearing and deterministic sensing signals, this paper addressed unauthorized-UAV detection over an entire surveillance region. Second, in addition to the composite-signal setting, we also investigate a special case where  the BS transmits only Gaussian-information-bearing signals for both sensing and communications. These extensions broaden the system scope and deepen the insight into the sensing-communication trade-off. The main contributions of this paper are summarized as follows: 
\begin{itemize}
		\item First, we jointly utilize both deterministic sensing and Gaussian -information-bearing signals for effective target detection. By applying the Neyman‑Pearson (NP) theorem, we derive the optimal detector, which maximizes the detection probability subject to a prescribed false‑alarm probability constraint. Leveraging the proposed detector, we analytically derive a closed-form expression that characterizes the detection probability as a function of the target false-alarm probability.
		\item Next, we optimize the transmit beamforming design at the BS to maximize the derived detection probability, subject to a minimum signal-to-interference-plus-noise ratio (SINR) requirement at the authorized UAV and a maximum transmit power constraint at the BS. When the BS transmits both types of signals, the detection probability exhibits a highly intricate and non-convex structure, rendering direct optimization intractable. To address this, we derive an accurate approximation of the detection probability assuming a sufficiently long sensing duration.
		By adopting the approximated detection probability as the objective function, the resulting non-convex beamforming design problem is formulated into a tractable form and efficiently tackled by adopting semi-definite relaxation (SDR) and successive convex approximation (SCA) techniques. 
		\item Moreover, we consider a special case where the BS relies solely on Gaussian-information-bearing signals to concurrently fulfill both sensing and communication functionalities. In this scenario, the SNR-constrained detection probability maximization problem remains inherently non-convex.  To address this challenge, we develop an SDR-based transmit beamforming design framework. Furthermore, we present a closed-form solution for the optimal beamforming vector when the intended detection location is known {\it a priori}, providing valuable insights into the system design and its performance limits.
		\item Finally, simulation results demonstrate that the joint exploitation of deterministic and Gaussian signals substantially improves the successful detection probability and validates the feasibility of solely relying on Gaussian-information-bearing signals for target detection. The results also show that the proposed beamforming designs significantly extend the achievable ISAC performance boundary, particularly under stringent communication-rate requirements.
\end{itemize}

\textit{Notations:} 
Boldface lowercase and uppercase letters represent vectors and matrices, respectively. For a square matrix $\mathbf S$, $\mathbf S^{-1}$, $\mathrm{tr}(\mathbf S)$, and $\det(\mathbf S)$ denote its inverse, trace, and determinant, respectively. $\mathbf S \succeq \mathbf{0}$ indicates that $\mathbf S$ is positive semi-definite. For any matrix $\mathbf X$, $\mathbf X^*$, $\mathbf X^{T}$, $\mathbf X^{H}$, $\mathrm{tr}(\mathbf X)$, and $\mathrm{rank} (\mathbf X)$ denote the conjugate, transpose, conjugate transpose, trace, and rank of a matrix $\mathbf X$, respectively. $\mathcal{N}(\mathbf{x}, \mathbf{\Sigma})$ and $\mathcal{C N}(\mathbf{x}, \mathbf{\Sigma})$ denote the real-valued Gaussian and circularly symmetric complex Gaussian distributions with mean vector $\mathbf x$ and covariance matrix $\mathbf \Sigma$, respectively. $\mathbb{C}^{x \times y}$ denotes the space of $x \times y$ complex matrices. 
The imaginary unit is denoted by $\jmath = \sqrt{-1}$. The real part of a complex number $x$ is denoted by $\mathrm{Re}\{x\}$. $\|\cdot\|$ and $|\cdot|$ denote the Euclidean norm of a vector and the modulus of a complex number, respectively. The operator $\mathbb E[\cdot]$ denotes statistical expectation, and $\otimes$ denotes the Kronecker product. The notation $\mathrm{Pr}\left\{X\ge x\right\}$ represents the probability that the random variable $X$ takes a value greater than or equal to $x$. The function $\mathcal{Q}_{\chi^2_{\nu}(\lambda)}(x)=\int_{x}^\infty\frac{1}{2}\left(\frac{x}{\lambda}\right)^{\frac{\nu-2}{4}}\exp\left[-\frac{1}{2}(x+\lambda)\right] I_{\frac{\nu}{2}-1}(\sqrt{\lambda x}) dx$ denotes the right-tail probability of a non-central chi-squared distribution with $\nu$ degrees of freedom and non-centrality parameter $\lambda$, where $I_m (\cdot)$ is the modified Bessel function of the first kind of order $m$. Finally, $Q(\cdot)$ denotes the Gaussian Q-function.

\section{System Model and Problem Formulation}\label{sec:system_model}
\begin{figure}[t]       
        \centering
        \includegraphics[width=0.42\textwidth]{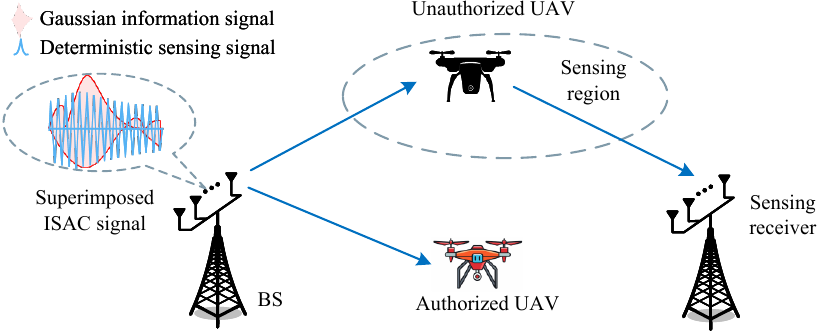}
        \vspace{-5pt}\caption{Illustration of a bistatic downlink ISAC system that simultaneously communicates with an authorized UAV and detects a potential unauthorized UAV in the surveillance region.}
        \label{fig:system_model}\vspace{-10pt}
\end{figure}
As shown in Fig.~\ref{fig:system_model}, we consider a bistatic downlink ISAC system, where a BS equipped with $M_\mathrm{t}$ transmit antennas simultaneously provides communication and sensing functionalities. Specifically, the BS delivers information to an authorized single-antenna UAV, while a dedicated sensing receiver equipped with $M_\mathrm{r}$ receive antennas monitors the presence of a potential unauthorized UAV in the surveillance region.\footnote{To facilitate the derivation of a closed-form detection probability expression and provide more insightful beamforming design results, we consider a basic setup with only a single authorized UAV for communications and a potential unauthorized UAV for detection. Nevertheless, the proposed detection-probability derivation framework can be extended to the more general case with multiple unauthorized UAVs, and the proposed beamforming design methods are also applicable to scenarios with multiple authorized UAVs, at the expenses of more involved notations.} 

\subsection{ISAC Signal Model}
First, we introduce the transmit signal model. The BS transmits a superimposed waveform containing deterministic sensing and Gaussian-information-bearing signals over the entire ISAC transmission. Specifically, the Gaussian-information-bearing signals are utilized for data transmission, whereas both signal components are jointly leveraged for target sensing to achieve improved detection performance. 

The set of time slots is denoted by $\mathcal L=\{1,\cdots,L\}$, where $L$ represents the total number of ISAC symbols. Let $s(l)\sim\mathcal{CN}(0,1)$ denote the Gaussian-information-bearing signal intended for the authorized UAV at time slot $l\in\mathcal{L}$. Let $\mathbf w \in \mathbb C^{M_\mathrm{t}\times1}$ denote the  beamforming vector for the authorized UAV  and $\mathbf x_0(l)\in \mathbb C^{M_\mathrm{t}\times1}$ denote the deterministic sensing signal at time slot $l\in\mathcal{L}$, respectively. Then, the sample covariance matrix of the deterministic sensing signals is given by
\begin{equation}\mathbf R_0=\frac{1}{L}\sum_{l=1}^L \mathbf x_0(l)\mathbf x_0^H(l).\end{equation}
Accordingly, the superimposed ISAC signal transmitted by the BS at time slot $l$ is expressed as
\begin{equation}\label{eq:transmit_signal}
\mathbf x(l) = \mathbf w s(l) + \mathbf x_0(l),\quad l\in\mathcal{L}.
\end{equation}
Note that for the superimposed signal in \eqref{eq:transmit_signal}, the sensing receiver is assumed to have perfect knowledge of the deterministic sensing signal realizations, while only statistical knowledge of the Gaussian-information-bearing signals is available. Let $P$ denote the maximum transmit power at the BS. Then, the average transmit power constraint is given by
\begin{equation}
\|\mathbf w\|^2+ \mathrm{tr}(\mathbf R_0) \le P.
\end{equation}

\subsection{Communication Model}
Next, we present the communication model and associated performance metrics. Let $\mathbf h\in\mathbb C^{M_\mathrm{t}\times 1}$ denote the channel vector from the BS to the authorized UAV. The received signal at the authorized UAV is 
\begin{equation}\label{eq:received_signal_CU}
y_\mathrm{c}(l) =\mathbf h^H \mathbf w s(l) + \mathbf h^H \mathbf x_0(l) + n_\mathrm{c}(l),\quad l\in \mathcal{L},
\end{equation}
where $n_\mathrm{c}(l)\sim\mathcal{CN}(0,\sigma_\mathrm{c}^2)$ denotes the AWGN at the authorized UAV receiver with noise power $\sigma_\mathrm{c}^2$. The SINR and the resulting achievable communication rate at the authorized UAV are respectively expressed as
\begin{equation} 
\gamma =\frac{|\mathbf h^H \mathbf w|^2}{\mathbf h^H\mathbf R_0\mathbf h + \sigma_\mathrm{c}^2},
\end{equation}
\begin{equation} 
R = \log_2(1+\gamma).
\end{equation}

\subsection{Sensing Model}

Finally, we introduce the target detection model for an unauthorized UAV in a given surveillance region. In particular, the sensing receiver processes the echo signals reflected along the BS-target-receiver cascade link to determine the presence of an unauthorized UAV in the given region. To facilitate sensing performance analysis, we partition the entire surveillance region into $Q$ sampling points, which is denoted by $\mathcal {Q} = \{1,\cdots,Q\}$. In the following, we first introduce the detection task for a target at a given point $q \in \mathcal{Q}$, and then adopt the minimum detection probability over all sampling points in $\mathcal{Q}$ as the sensing performance metric. 

We assume that the sensing channel is dominated by a line-of-sight (LoS) link \cite{9696263,10879807}. In particular, let $\mathbf a_q\in\mathbb C^{M_\mathrm{t}\times1}$ and $\mathbf b_q\in \mathbb{C}^{M_\mathrm{r}\times 1}$ denote the transmit and receive steering vectors at the BS and the sensing receiver towards location $q \in \mathcal {Q}$, respectively, i.e.,
\begin{subequations}\label{eq:steering_vector}
\begin{align}
\mathbf a_q 
=& \left[1, e^{\frac{\jmath2\pi d_\mathrm{a}\sin\phi_q}{\lambda}},\cdots,e^{\frac{\jmath2\pi(M_\mathrm{t}-1)d_\mathrm{a}\sin\phi_q}{\lambda}}\right]^T,\\
\mathbf b_q
 =& \left[1, e^{\frac{\jmath 2\pi d_\mathrm{a}\sin\theta_q}{\lambda}},\cdots,e^{\frac{\jmath 2\pi(M_\mathrm{r}-1)d_\mathrm{a}\sin\theta_q}{\lambda}}\right]^T,
\end{align}
\end{subequations}
where $d_\mathrm{a}$ denotes the distance between two neighboring antennas, $\phi_q$ and $\theta_q$ are the direction-of-departure (DoD) and DoA of the $q$-th sampling with respect to the BS and sensing receiver, respectively.
Let $\alpha_q \in \mathbb C$ denote the complex channel coefficient of the BS-target-receiver cascade link with 
$|\alpha_q|^2 = \frac{\eta\sigma_\mathrm{t}}{d_{1,q}^{2}d_{2,q}^{2}}$\cite{richards2005fundamentals},
where $\eta = \frac{c^2}{64\pi^3f^2}$ is a constant determined by the speed of light $c$ and the waveform carrier frequency $f$, $\sigma_\mathrm{t}$ denotes the target’s RCS, and $d_{1,q}$ and $d_{2,q}$ denote the distance between the BS and the $q$-th sampling and that between the $q$-th sampling and the sensing receiver, respectively.\footnote{This paper focuses on determining the presence of a potential intruder within a certain region. Accordingly, the channel parameters $\alpha_q$, $\phi_q$, and $\theta_q$ for each hypothesized grid point are assumed to be known to facilitate the detector and beamforming design\cite{10380513}.}
Next, the sensing channel matrix corresponding to the BS-target-receiver cascade link is expressed as
$\mathbf H_\mathrm{s} = \alpha_q\mathbf b_q \mathbf a_q^T$\cite{richards2005fundamentals}.

Let  $\mathcal H_0$ and $\mathcal H_1$ denote the hypotheses corresponding to the absence and presence of the unauthorized UAV, respectively. 
The received echo signal at the sensing receiver for a target located at the $q$-th sampling is given by
\begin{subequations}
  \begin{align}\label{eq:echo_0}
&\mathcal{H}_0:\mathbf y_\mathrm{s}(l) = \mathbf n_\mathrm{s}(l),\\\label{eq:echo_1}
&\mathcal{H}_1:\mathbf y_\mathrm{s}(l) = \alpha_q\mathbf b_q \mathbf a_q^T\mathbf w s(l)+ \alpha_q\mathbf b_q \mathbf a_q^T \mathbf x_0(l) + \mathbf n_\mathrm{s}(l),
  \end{align}
\end{subequations}
respectively, where $\mathbf n_\mathrm{s}(l) \sim \mathcal{CN}(\mathbf 0,\sigma_\mathrm{s}^2\mathbf I_{M_\mathrm{r}})$ denotes the AWGN vector at the sensing receiver with noise power $\sigma_\mathrm{s}^2$ at each receive antenna. 

Let $P_{\mathrm{D},q}$ and $P_{\mathrm{FA},q}$ denote the detection and false-alarm probabilities for the $q$-th sampling location, respectively. Specifically, the detection probability characterizes the likelihood that the sensing receiver correctly identities the presence of a target. In particular, the false-alarm probability denotes how often it incorrectly declares a target when none is present. In general, there exists an inherent performance trade-off between detection and false-alarm probabilities. Thus, radar sensing systems widely adopt the detection probability achieved under a maximum predetermined false-alarm probability as the performance metric\cite{richards2005fundamentals,10380513}, as will be derived in Section~\ref{sec:detection_probability_derivation}. Furthermore, to account for the uncertainty of the unauthorized UAV, we adopt the minimum detection probability over all sampling points in the entire surveillance region as the  sensing performance metric.

\subsection{Problem Formulation}
In this paper, we focus on exploring the fundamental performance trade-off between sensing and communication. In particular, we aim to maximize the minimum detection probability over the entire surveillance region, by jointly optimizing the transmit beamforming $\mathbf w$ and $\mathbf R_0$, subject to a maximum transmit power constraint at the BS and a minimum communication SINR requirement at the authorized UAV, which is formulated as
\begin{subequations}
  \begin{align}\notag
   \text{(P1)}:\maximize_{\mathbf w,~\mathbf R_0}\quad \min_{q\in \mathcal Q} &\quad  \{P_{\mathrm{D},q}\}\\ \label{eq:st_SINR}
    \text { s.t. }& \quad \frac{|\mathbf h^H \mathbf w|^2}{\mathbf h^H\mathbf R_0\mathbf h + \sigma_\mathrm{c}^2}\ge \gamma_0,\\\label{eq:st_powe}
    &\quad \|\mathbf w\|^2+\mathrm{tr}(\mathbf R_0) \le P,\\\label{eq:st_semi}
    &\quad  \mathbf R_0 \succeq \mathbf 0,
  \end{align}
\end{subequations}
where $\gamma_0$ denotes the minimum SINR requirement for reliable communication with the authorized UAV.

\section{Target Detector Design and Detection Probability Derivation}\label{sec:detection_probability_derivation}

In this section, we first develop an NP-based detector that jointly exploits both Gaussian-information-bearing and deterministic sensing signals and then derive a closed-form expression for the corresponding detection probability $P_{\mathrm{D},q}$ at the $q$-th sampling location.
\subsection{Target Detector}
To facilitate the detector design and performance analysis, we stack the received signals in \eqref{eq:echo_0} and \eqref{eq:echo_1} over the entire observation interval $L$ into the following vectors: 
\begin{subequations}
\begin{align}
\tilde{\mathbf s} &= [\mathbf w^T s(1),\cdots,\mathbf w^T s(L)]^T\in \mathbb C^{M_\mathrm{t}L\times1},\\
\tilde{\mathbf x}_0 &= [\mathbf x_0^T(1),\cdots,\mathbf x_0^T(L)]^T\in \mathbb C^{M_\mathrm{t}L\times1},\\
\tilde{\mathbf y} &= [\mathbf y_\mathrm{s}^T(1),\cdots,\mathbf y_\mathrm{s}^T(L)]^T\in \mathbb C^{M_\mathrm{r}L\times1},\\
\tilde{\mathbf n}_\mathrm{s} &= [\mathbf n_\mathrm{s}^T(1),\cdots,\mathbf n_\mathrm{s}^T(L)]^T\in \mathbb C^{M_\mathrm{r}L\times1}.
\end{align}
\end{subequations}
Then, the received echo signals over the whole sensing interval under hypotheses $\mathcal H_0$ and $\mathcal H_1$ are respectively given as
\begin{subequations}
  \begin{align}\label{eq:received_signal_I}
&\mathcal{H}_0:\tilde{\mathbf y}= \tilde{\mathbf n}_\mathrm{s},\\\label{eq:received_signal_II}
&\mathcal{H}_1:\tilde{\mathbf y} = \underbrace{\left(\mathbf I_L \otimes \alpha_q\mathbf b_q \mathbf a_q^T\right) \tilde{\mathbf s}}_{\mathbf u_1:~\text{Random signal}} + \underbrace{\left(\mathbf I_L \otimes \alpha_q\mathbf b_q \mathbf a_q^T\right) \tilde{\mathbf x}_0}_{\mathbf u_2:~\text{Deterministic signal}} + \tilde{\mathbf n}_\mathrm{s}.
\end{align}
\end{subequations}
At the sensing receiver, the term $\mathbf u_2 = \left(\mathbf I_L \otimes \alpha_q\mathbf b_q \mathbf a_q^T\right) \tilde{\mathbf x}_0$ is deterministic and known exactly, while $\mathbf u_1=\left(\mathbf I_L \otimes \alpha_q\mathbf b_q \mathbf a_q^T\right) \tilde{\mathbf s}$ is random with zero mean and characterized by the covariance matrix 
\begin{equation}
\mathbf {C}_{q} = \mathbb E \left[\mathbf u_1\mathbf u_1^H\right]=\mathbf I_L \otimes \left(|\alpha_q|^2|\mathbf a_q^T \mathbf w|^2\mathbf b_q \mathbf b_q^H\right), \forall q\in \mathcal Q.
\end{equation}
Based on the received signal models in \eqref{eq:received_signal_I} and \eqref{eq:received_signal_II}, we aim to determine the presence of the unauthorized UAV. 
Under the two hypotheses, the likelihood functions of the received signal $\tilde{\mathbf y}$ are respectively given in \eqref{eq:PDF_I} and \eqref{eq:PDF_II} at the top of this page.
\begin{figure*}
\begin{subequations}
  \begin{align}\label{eq:PDF_I}
&p(\tilde{\mathbf y};\mathcal{H}_0) = \frac{\exp\left(-\frac{1}{\sigma_\mathrm{s}^2}\tilde{\mathbf y}^H\tilde{\mathbf y}\right)}{\pi^{M_\mathrm{r}L}\det(\sigma_\mathrm{s}^2\mathbf I_{M_\mathrm{r}L})},\\\label{eq:PDF_II}
&p(\tilde{\mathbf y};\mathcal{H}_1) = \frac{\exp\left(-(\tilde{\mathbf y} - \left(\mathbf I_L \otimes \alpha_q\mathbf b_q \mathbf a_q^T\right) \tilde{\mathbf x}_0)^H(\mathbf {C}_{q}+\sigma_\mathrm{s}^2\mathbf I_{M_\mathrm{r}L})^{-1}(\tilde{\mathbf y} - \left(\mathbf I_L \otimes \alpha_q\mathbf b_q \mathbf a_q^T\right) \tilde{\mathbf x}_0)\right)}{\pi^{M_\mathrm{r}L}\det(\mathbf {C}_{q}+\sigma_\mathrm{s}^2\mathbf I_{M_\mathrm{r}L})}.
\end{align}
\end{subequations}
\hrulefill
\end{figure*}

Next, we introduce the optimal target detector by utilizing the NP theorem\cite{richards2005fundamentals}. In particular, the optimal detector is obtained by comparing the ratio of the likelihood functions under hypotheses $\mathcal H_1$ and $\mathcal H_0$ with a decision threshold $\delta$, i.e., 
\begin{equation}\label{eq:NP_dector} 
\frac{p(\tilde{\mathbf y};\mathcal{H}_1)}{p(\tilde{\mathbf y};\mathcal{H}_0)} \stackrel{\mathcal{H}_1}{\underset{\mathcal{H}_0}{\gtrless}} \delta,
\end{equation}
where $\delta$ is selected to satisfy a specified false-alarm probability. By substituting the likelihood functions in \eqref{eq:PDF_I} and \eqref{eq:PDF_II} into \eqref{eq:NP_dector} and taking the logarithm of both sides of \eqref{eq:NP_dector}, the NP-based detector $T(\tilde{\mathbf y})$ is given as 
\begin{equation}\label{eq:detector_re}
\begin{split}
T(\tilde{\mathbf y}) &\triangleq \tilde{\mathbf y}^H\left(\frac{1}{\sigma_\mathrm{s}^2}\mathbf I_{M_\mathrm{r}L}-(\mathbf {C}_{q}+\sigma_\mathrm{s}^2\mathbf I_{M_\mathrm{r}L})^{-1}\right)\tilde{\mathbf y}\\
&\quad+2\mathrm{Re}\left\{\left(\left(\mathbf I_L\otimes\alpha_q\mathbf b_q \mathbf a_q^T\right)\tilde{\mathbf x}_0\right)^H(\mathbf {C}_{q}+\sigma_\mathrm{s}^2\mathbf I_{M_\mathrm{r}L})^{-1}\tilde{\mathbf y}\right\}\\
&\stackrel{\mathcal{H}_1}{\underset{\mathcal{H}_0}{\gtrless}} \delta',
\end{split}
\end{equation}
where $\delta' = \ln\delta + \left(\left(\mathbf I_L \otimes\alpha_q\mathbf b_q \mathbf a_q^T\right)\tilde{\mathbf x}_0\right)^H(\mathbf {C}_{q}+\sigma_\mathrm{s}^2\mathbf I_{M_\mathrm{r}L})^{-1}\left(\left(\mathbf I_L \otimes\alpha_q\mathbf b_q \mathbf a_q^T\right)\tilde{\mathbf x}_0\right)-M_\mathrm{r}L\ln \sigma_\mathrm{s}^{2}+\ln \det(\mathbf {C}_{q}+\sigma_\mathrm{s}^2\mathbf I_{M_\mathrm{r}L})$ is the equivalent detection threshold after logarithmic transformation. By utilizing the Sherman-Morrison-Woodbury identity\cite{sherman1950adjustment}, we have 
\begin{equation}\label{eq:C_inv}
\begin{split}
&\quad ~(\mathbf {C}_{q}+\sigma_\mathrm{s}^2\mathbf I_{M_\mathrm{r}L})^{-1} \\
&=\left(\mathbf I_L \otimes \left(|\alpha_q|^2|\mathbf a_q^T \mathbf w|^2\mathbf b_q \mathbf b_q^H\right)+\sigma_\mathrm{s}^2\mathbf I_{M_\mathrm{r}L}\right)^{-1}\\
&=\frac{1}{\sigma_\mathrm{s}^2}\mathbf I_L \otimes \left(\mathbf I_{M_\mathrm{r}} - \frac{|\alpha_q|^2|\mathbf a_q^T \mathbf w|^2|\mathbf b_q \mathbf b_q^H/\sigma_\mathrm{s}^2}{1+\gamma_{\mathrm{c},q}}\right),
\end{split}
\end{equation}
where 
$
\gamma_{\mathrm{c},q} = \frac{|\alpha_q\mathbf b_q \mathbf a_q^T\mathbf w|^2}{
\sigma_\mathrm{s}^2}
=\frac{|\alpha_q|^2|\mathbf a_q^T\mathbf w|^2\|\mathbf b_q\|^2}{\sigma_\mathrm{s}^2}
= \frac{|\alpha_q|^2|\mathbf a_q^T\mathbf w|^2M_\mathrm{r}}{\sigma_\mathrm{s}^2}
$
denotes the ratio between the power of the received Gaussian signal and the noise at the sensing receiver.
Then, by substituting \eqref{eq:C_inv} into \eqref{eq:detector_re}, the test statistic $T(\tilde{\mathbf y})$ is further expressed as
\begin{equation}\label{eq:detector_I}
\begin{split}
T(\tilde{\mathbf y}) 
&=  \frac{|\alpha_q|^2|\mathbf a_q^T \mathbf w|^2}{\sigma_\mathrm{s}^4(1+\gamma_{\mathrm{c},q})}\sum_{l=1}^L \left|\mathbf b_q^H \mathbf y_\mathrm{s}(l)\right|^2\\
&\quad +\! \frac{2}{\sigma_\mathrm{s}^2(1+\gamma_{\mathrm{c},q})}\mathrm{Re}\!\left\{\sum_{l=1}^L \alpha_q^*\mathbf x_0^H(l) \mathbf a_q^*\mathbf b_q^H \mathbf y_\mathrm{s}(l)\right\}\\
&=  \frac{\gamma_{\mathrm{c},q}}{M_\mathrm{r}\sigma_\mathrm{s}^2(1+\gamma_{\mathrm{c},q})}\sum_{l=1}^L \left|\mathbf b_q^H \mathbf y_\mathrm{s}(l)\right|^2\\
 &\quad+\! \frac{2}{\sigma_\mathrm{s}^2(1+\gamma_{\mathrm{c},q})}\mathrm{Re}\!\left\{\sum_{l=1}^L \alpha_q^*\mathbf x_0^H(l) \mathbf a_q^*\mathbf b_q^H \mathbf y_\mathrm{s}(l)\right\}.
\end{split}
\end{equation}
It is evident that the NP-based detector in \eqref{eq:detector_I} consists of two components: an energy-based detector that captures the energy of received Gaussian information signal $\mathbf u_1$ and a matched-filter detector that correlates the received echo signal with the deterministic signal $\mathbf u_2$.

\begin{remark}
When the BS transmits only Gaussian-information-bearing signals, the NP-based detector in \eqref{eq:detector_I} reduces to 
\begin{equation}\label{eq:detector_Gaussian}
T_\text{ran}(\tilde{\mathbf y}) 
=\frac{\gamma_{\mathrm{c},q}}{M_\mathrm{r}\sigma_\mathrm{s}^2(1+\gamma_{\mathrm{c},q})}\sum_{l=1}^L \left|\mathbf b_q^H \mathbf y_\mathrm{s}(l)\right|^2 \stackrel{\mathcal{H}_1}{\underset{\mathcal{H}_0}{\gtrless}} \delta',
\end{equation}
which is an energy-based detector that compares the received signal energy against the decision threshold. 
\end{remark}

\subsection{Detection and False-Alarm Probabilities}
In the following, we first derive closed-form expressions for the detection and false-alarm probabilities under an arbitrary detection threshold, and then present a closed-form expression for the detection probability under a prescribed  false-alarm probability. 
To facilitate these derivations, we reformulate the test statistic in \eqref{eq:detector_I} as 
\begin{equation}\label{eq:detector_II}
\begin{split}
T(\tilde{\mathbf y}) 
&=  \frac{\gamma_{\mathrm{c},q}}{M_\mathrm{r}\sigma_\mathrm{s}^2(1+\gamma_{\mathrm{c},q})}\left(\sum_{l=1}^L \left|\mathbf b_q^H \mathbf y_\mathrm{s}(l)\right|^2 \right.\\
 &\quad \left.+ \frac{2M_\mathrm{r}}{\gamma_{\mathrm{c},q}}\mathrm{Re}\left\{\sum_{l=1}^L \alpha_q^*\mathbf x_0^H(l) \mathbf a_q^*\mathbf b_q^H \mathbf y_\mathrm{s}(l)\right\}\right)\\
 &\!=\!\frac{\gamma_{\mathrm{c},q}}{M_\mathrm{r}\sigma_\mathrm{s}^2(1+\gamma_{\mathrm{c},q})}\sum_{l=1}^L\left|\mathbf b_q^H\mathbf y_\mathrm{s}(l)+\frac{M_\mathrm{r}}{\gamma_{\mathrm{c},q}}\alpha_q\mathbf a_q^T\mathbf x_0(l)\right|^2\\
&\quad - \frac{\gamma_{\mathrm{c},q}}{M_\mathrm{r}\sigma_\mathrm{s}^2(1+\gamma_{\mathrm{c},q})}\sum_{l=1}^L\left|\frac{M_\mathrm{r}}{\gamma_{\mathrm{c},q}}\alpha_q\mathbf a_q^T\mathbf x_0(l)\right|^2\\
&\!=\!\frac{\gamma_{\mathrm{c},q}}{M_\mathrm{r}\sigma_\mathrm{s}^2(1+\gamma_{\mathrm{c},q})}\sum_{l=1}^L\left|\mathbf b_q^H\mathbf y_\mathrm{s}(l)+\frac{M_\mathrm{r}}{\gamma_{\mathrm{c},q}}\alpha_q\mathbf a_q^T\mathbf x_0(l)\right|^2\\
&\quad - \frac{L\gamma_{\mathrm{s},q}}{(1+\gamma_{\mathrm{c},q})\gamma_{\mathrm{c},q}},
\end{split}
\end{equation}
where 
$\gamma_{\mathrm{s},q} =\frac{\frac{1}{L}\sum_{l=1}^L|\alpha_q\mathbf b_q \mathbf a_q^T \mathbf x_0(l)|^2}{\sigma_\mathrm{s}^2}
=\frac{|\alpha_q|^2\mathbf a_q^T\mathbf R_0\mathbf a_q^*\|\mathbf b_q\|^2}{\sigma_\mathrm{s}^2}
= \frac{|\alpha_q|^2\mathbf a_q^T\mathbf R_0\mathbf a_q^*M_\mathrm{r}}{\sigma_\mathrm{s}^2}
$ denotes the ratio between the received deterministic signal power and the noise at the sensing receiver.\footnote{To facilitate analytical derivations in the sequel, we assume that $\gamma_{\mathrm{c},q}>0$, which differentiates the general case from the deterministic-signal-only scenario.}

Next, we consider the case where the unauthorized UAV is absent and derive the false-alarm probability. Under hypothesis $\mathcal {H}_0$, the test statistic in \eqref{eq:detector_II} becomes
\begin{equation}
\begin{split}
T(\tilde{\mathbf y};\mathcal{H}_0) 
&=\frac{\gamma_{\mathrm{c},q}}{M_\mathrm{r}\sigma_\mathrm{s}^2(1+\gamma_{\mathrm{c},q})}\sum_{l=1}^L\left|\mathbf b_q^H\mathbf n_\mathrm{s}(l)+\frac{M_\mathrm{r}}{\gamma_{\mathrm{c},q}}\alpha_q\mathbf a_q^T\mathbf x_0(l)\right|^2\\
&\quad - \frac{L\gamma_{\mathrm{s},q}}{(1+\gamma_{\mathrm{c},q})\gamma_{\mathrm{c},q}}.
\end{split}
\end{equation}
It is clear that $\mathbf b_q^H\mathbf n_\mathrm{s}(l)+\frac{M_\mathrm{r}}{\gamma_{\mathrm{c},q}}\alpha_q\mathbf a_q^T\mathbf x_0(l) \in \mathbb C$ follows a CSCG distribution with mean $\frac{M_\mathrm{r}}{\gamma_{\mathrm{c},q}}\alpha_q\mathbf a_q^T\mathbf x_0(l)$ and variance $M_\mathrm{r}\sigma_\mathrm{s}^2$. Thus, we define a normalized test variable
\begin{equation}
\begin{split}
T'(\tilde{\mathbf y};\mathcal{H}_0)&=\left(T(\tilde{\mathbf y};\mathcal{H}_0) + \frac{L\gamma_{\mathrm{s},q}}{(1+\gamma_{\mathrm{c},q})\gamma_{\mathrm{c},q}}\right)\frac{2(1+\gamma_{\mathrm{c},q})}{\gamma_{\mathrm{c},q}}\\
&=\frac{2(1+\gamma_{\mathrm{c},q})}{\gamma_{\mathrm{c},q}}T(\tilde{\mathbf y};\mathcal{H}_0) + \frac{2L\gamma_{\mathrm{s},q}}{\gamma_{\mathrm{c},q}^2}\\
&=\frac{2}{M_\mathrm{r}\sigma_\mathrm{s}^2}\sum_{l=1}^L\left|\mathbf b_q^H\mathbf n_\mathrm{s}(l)+\frac{M_\mathrm{r}}{\gamma_{\mathrm{c},q}}\alpha_q\mathbf a_q^T\mathbf x_0(l)\right|^2,
 \end{split}
\end{equation}
which follows a non-central chi-squared distribution with $\nu_1=2L$ degrees of freedom and non-centrality parameter 
\begin{equation}
\begin{split}
\lambda_1 &= \frac{2}{M_\mathrm{r}\sigma_\mathrm{s}^2}\sum_{l=1}^L\left|\frac{M_\mathrm{r}}{\gamma_{\mathrm{c},q}}\alpha_q\mathbf a_q^T\mathbf x_0(l)\right|^2\\
&=\frac{2LM_\mathrm{r}|\alpha_q|^2\mathbf a_q^T\mathbf R_0\mathbf a_q^*}{\sigma_\mathrm{s}^2\gamma_{\mathrm{c},q}}
 = \frac{2L\gamma_{\mathrm{s},q}}{\gamma_{\mathrm{c},q}^2}.
\end{split}
\end{equation}
Then, the false-alarm probability under a given detection threshold $\delta'$ is 
\begin{equation}\label{eq:P_FA_delta}
\begin{split}
 P_{\text{FA},q} &= \mathrm{Pr}\left\{T(\tilde{\mathbf y};\mathcal{H}_0)\ge \delta'\right\}\\
 & = \mathrm{Pr}\left\{T'(\tilde{\mathbf y};\mathcal{H}_0)\ge \frac{2(1+\gamma_{\mathrm{c},q})\delta'}{\gamma_{\mathrm{c},q}} + \frac{2L\gamma_{\mathrm{s},q}}{\gamma_{\mathrm{c},q}^2}\right\}\\
 &= \mathcal{Q}_{\chi^2_{2L}(\lambda_1)}\left(\frac{2(1+\gamma_{\mathrm{c},q})\delta'}{\gamma_{\mathrm{c},q}} + \frac{2L\gamma_{\mathrm{s},q}}{\gamma_{\mathrm{c},q}^2}\right).
 \end{split}
\end{equation}

Next, we consider the case where the unauthorized UAV is present and derive the corresponding detection probability. Under the alternative hypothesis $\mathcal H_1$, the test statistic in \eqref{eq:detector_II} becomes 
\begin{equation}
\begin{split}
T(\tilde{\mathbf y};\mathcal{H}_1) 
&=\frac{\gamma_{\mathrm{c},q}}{M_\mathrm{r}\sigma_\mathrm{s}^2(1+\gamma_{\mathrm{c},q})}\sum_{l=1}^L\left|\mathbf b_q^H\left(\alpha_q\mathbf b_q \mathbf a_q^T\mathbf w s(l)\right.\right.\\
&\quad\left.\left.+ \alpha_q\mathbf b_q \mathbf a_q^T \mathbf x_0(l)+\mathbf n_\mathrm{s}(l)\right)+\frac{M_\mathrm{r}}{\gamma_{\mathrm{c},q}}\alpha_q\mathbf a_q^T\mathbf x_0(l)\right|^2\\
&\quad - \frac{L\gamma_{\mathrm{s},q}}{(1+\gamma_{\mathrm{c},q})\gamma_{\mathrm{c},q}}.
\end{split}
\end{equation}
It can be observed that $\mathbf b_q^H\left(\alpha_q\mathbf b_q \mathbf a_q^T\mathbf w s(l)+ \alpha_q\mathbf b_q \mathbf a_q^T \mathbf x_0(l)+\mathbf n_\mathrm{s}(l)\right)+\frac{M_\mathrm{r}}{\gamma_{\mathrm{c},q}}\alpha_q\mathbf a_q^T\mathbf x_0(l)$ follows a CSCG distribution with mean $\frac{1+\gamma_{\mathrm{c},q}}{\gamma_{\mathrm{c},q}}M_\mathrm{r}\alpha_q\mathbf a_q^T \mathbf x_0(l)$ and variance $M_\mathrm{r}\sigma_\mathrm{s}^2(1+\gamma_{\mathrm{c},q})$. Then, the corresponding normalized test variable is defined as
\begin{equation}
\begin{split}
T'(\tilde{\mathbf y};\mathcal{H}_1)&=\left(T(\tilde{\mathbf y};\mathcal{H}_1) + \frac{L\gamma_{\mathrm{s},q}}{(1+\gamma_{\mathrm{c},q})\gamma_{\mathrm{c},q}}\right)\frac{2}{\gamma_{\mathrm{c},q}}\\
&=\frac{2}{\gamma_{\mathrm{c},q}}T(\tilde{\mathbf y};\mathcal{H}_1) + \frac{2L\gamma_{\mathrm{s},q}}{(1+\gamma_{\mathrm{c},q})\gamma_{\mathrm{c},q}^2}\\
&=\frac{2}{M_\mathrm{r}\sigma_\mathrm{s}^2(1+\gamma_{\mathrm{c},q})}\sum_{l=1}^L\left|\mathbf b_q^H\left(\alpha_q\mathbf b_q \mathbf a_q^T\mathbf w s(l)\right.\right.\\
&\quad\left.\left.+ \alpha_q\mathbf b_q \mathbf a_q^T \mathbf x_0(l)+\mathbf n_\mathrm{s}(l)\right)+\frac{M_\mathrm{r}}{\gamma_{\mathrm{c},q}}\alpha_q\mathbf a_q^T\mathbf x_0(l)\right|^2,
 \end{split}
\end{equation}
which follows a non-central chi-squared distribution with $\nu_2=2L$ degrees of freedom and non-centrality parameter 
\begin{equation}
\begin{split}
\lambda_2 & =\frac{2}{M_\mathrm{r}\sigma_\mathrm{s}^2(1+\gamma_{\mathrm{c},q})}\sum_{l=1}^L\left|\mathbf b_q^H\alpha_q\mathbf b_q \mathbf a_q^T \mathbf x_0(l)\!+\!\frac{M_\mathrm{r}}{\gamma_{\mathrm{c},q}}\alpha_q\mathbf a_q^T\mathbf x_0(l)\right|^2\\
&= \frac{2L|\alpha_q|^2 M_\mathrm{r}\mathbf a_q^T\mathbf R_0 \mathbf a_q^*(1+\gamma_{\mathrm{c},q})}{\sigma_\mathrm{s}^2\gamma_{\mathrm{c},q}^2}\\
&= \frac{2L\gamma_{\mathrm{s},q}(1+\gamma_{\mathrm{c},q})}{\gamma_{\mathrm{c},q}^2}\\
&=\lambda_1(1+\gamma_{\mathrm{c},q}).
\end{split}
\end{equation}
Thus, the detection probability under a given detection threshold $\delta'$ is  given as
\begin{equation}\label{eq:P_D_delta}
\begin{split}
P_{\text{D},q} &= \mathrm{Pr}\left\{T(\tilde{\mathbf y};\mathcal{H}_1)\ge \delta'\right\}\\
& = \mathrm{Pr}\left\{T'(\tilde{\mathbf y};\mathcal{H}_1)\ge \frac{2\delta'}{\gamma_{\mathrm{c},q}} + \frac{2L\gamma_{\mathrm{s},q}}{(1+\gamma_{\mathrm{c},q})\gamma_{\mathrm{c},q}^2}\right\}\\
&= \mathcal{Q}_{\chi^2_{2L}(\lambda_2)}\left(\frac{2\delta'}{\gamma_{\mathrm{c},q}} + \frac{2L\gamma_{\mathrm{s},q}}{(1+\gamma_{\mathrm{c},q})\gamma_{\mathrm{c},q}^2}\right).
 \end{split}
\end{equation}

Finally, based on the false-alarm probability in \eqref{eq:P_FA_delta} and the detection probability in \eqref{eq:P_D_delta}, we establish the following result.
\begin{proposition}
The detection probability $ P_{\text{D},q}$ under a predetermined false-alarm probability $P_{\text{FA},q}$ is 
\begin{equation}\label{eq:PD_FA}
 P_{\text{D},q} = \mathcal{Q}_{\chi^2_{2L}\left(\frac{2L\gamma_{\mathrm{s},q}}{\gamma_{\mathrm{c},q}^2}(1+\gamma_{\mathrm{c},q})\right)}\left(\frac{\mathcal{Q}^{-1}_{\chi^2_{2L}\left(\frac{2L\gamma_{\mathrm{s},q}}{\gamma_{\mathrm{c},q}^2}\right)}\left(P_{\text{FA},q}\right)}{1+\gamma_{\mathrm{c},q}}\right).
\end{equation}
\end{proposition}
\begin{remark}\label{re:P_D_random}
When the BS transmits only Gaussian-information-bearing signals for both sensing and communication, the detection probability in \eqref{eq:PD_FA} becomes
\begin{equation}\label{eq:PD_FA_Gaussian}
 P'_{\text{D},q} = \mathcal{Q}_{\chi^2_{2L}\left(0\right)}\left(\frac{\mathcal{Q}^{-1}_{\chi^2_{2L}\left(0\right)}\left(P_{\text{FA},q}\right)}{1+\gamma_{\mathrm{c},q}}\right),
\end{equation}
which is monotonically increasing with respect to $\gamma_{\mathrm{c},q}$.
\end{remark}

\section{Transmit Beamforming Design With Both Gaussian Information and Deterministic Sensing Signals}\label{sec:BF_design}

In this section, we jointly optimize the transmit beamforming vector $\mathbf w$ and covariance matrix $\mathbf R_0$ at the BS to maximize the minimum detection probability of an unauthorized UAV over the entire surveillance region, subject to a minimum communication SINR requirement at the authorized UAV and a total transmit power budget at the BS.
\subsection{Detection Probability Approximation}

First, the exact detection probability expression in \eqref{eq:PD_FA} is mathematically intractable for direct beamforming optimization, due to the complex integral form of the involved non-central chi-squared right-tail probability function. To facilitate the transmit beamforming design, we introduce an asymptotic detection probability approximation for sufficiently large sensing duration $L$.
\begin{proposition}\label{prop:approximation_L}
For sufficiently large sensing time $L$, the detection probability in \eqref{eq:PD_FA} is approximated as\footnote{The accuracy of the proposed approximation in Proposition~\ref{prop:approximation_L} will be verified in the simulation results.}
\begin{subequations}
\begin{align}\label{PD_approximation_a}
&P_{\mathrm{D},q} \\
\stackrel{(a_{1})}{\approx}&  Q\left(\frac{Q^{-1}(P_{\text{FA},q})\sqrt{\gamma_{\mathrm{c},q}^2\!+\!2\gamma_{\mathrm{s},q}}\!-\!\sqrt{L}\left(\gamma_{\mathrm{c},q}^2\!+\!\gamma_{\mathrm{s},q}(2\!+\!\gamma_{\mathrm{c},q})\right)}{(1+\gamma_{\mathrm{c},q})\sqrt{\gamma_{\mathrm{c},q}^2+2\gamma_{\mathrm{s},q}(1+\gamma_{\mathrm{c},q})}}\right)\\\label{PD_approximation_b}
\stackrel{(a_{2})}{\approx}&  Q\left(Q^{-1}(P_{\text{FA},q})-\sqrt{L}\sqrt{\gamma_{\mathrm{c},q}^2+2\gamma_{\mathrm{s},q}}\right)\\
\triangleq~& \tilde P_{\mathrm{D},q}. 
\end{align}
\end{subequations}
\begin{IEEEproof}
Detailed proof is given in Appendix~\ref{app:proof_of_the_approximation}.
\end{IEEEproof}

\end{proposition}

\subsection{Transmit Beamforming Design with Approximated Detection Probability}
Next, by applying the approximations obtained in Proposition~\ref{prop:approximation_L}, maximizing the detection probability is equivalent to maximizing $\gamma_{\mathrm{c},q}^2+2\gamma_{\mathrm{s},q}$. As such, the SINR-constrained detection probability maximization problem can be reformulated as
\begin{subequations}\nonumber
  \begin{align}
   \text{(P2)}:\maximize_{\mathbf w,~\mathbf R_0} ~ \min_{q\in \mathcal Q}  &~  \{M_\mathrm{r}|\alpha_q|^4|\mathbf a_q^T \mathbf w|^4/\sigma_\mathrm{s}^2+2|\alpha_q|^2\mathbf a_q^T \mathbf R_0 \mathbf a_q^*\}\\ 
    \text { s.t. }& \quad \eqref{eq:st_SINR},~\eqref{eq:st_powe},~\text{and}~\eqref{eq:st_semi}.
  \end{align}
\end{subequations}
Then, we employ a combination of SDR and SCA techniques to obtain a high-quality solution to (P2). By defining $\mathbf W = \mathbf w\mathbf w^H$ with $\mathbf W \succeq \mathbf 0$ and $\mathrm{rank}(\mathbf W)=1$ and introducing an auxiliary optimization variable $t$, problem (P2) is equivalently expressed as
\setcounter{equation}{30} 
\begin{subequations}
  \begin{align}\notag
   &\text{(P2.1)}:\maximize_{\mathbf W,~\mathbf R_0,~t} \quad  t\\\notag
   &\text {s.t.} \quad M_\mathrm{r}|\alpha_q|^4\mathrm{tr}^2(\mathbf W \mathbf a_q^*\mathbf a_q^T)/\sigma_\mathrm{s}^2+2|\alpha_q|^2\mathrm{tr}(\mathbf R_0 \mathbf a_q^*\mathbf a_q^T) \ge t,\\\label{eq:st_PD}
   &\quad \quad \forall q\in \mathcal Q,\\ \label{eq:st_SINR_W}
    & \quad \quad \mathrm{tr}(\mathbf W \mathbf h\mathbf h^H)\ge \gamma_0\mathrm{tr}(\mathbf R_0\mathbf h\mathbf h^H) + \gamma_0\sigma_\mathrm{c}^2, \\\label{eq:st_powe_W}
    &\quad \quad \mathrm{tr}(\mathbf W)+\mathrm{tr}(\mathbf R_0) \le P,\\\label{eq:st_semi_W}
    &\quad \quad \mathbf W \succeq \mathbf 0, \mathbf R_0 \succeq \mathbf 0,\\ \label{eq:st_rank_W}
    &\quad \quad \mathrm{rank}(\mathbf W)=1.
  \end{align}
\end{subequations}
Problem (P2.1) is non-convex due to the quartic term in \eqref{eq:st_PD} and the rank-one constraint \eqref{eq:st_rank_W}. Dropping the rank‑one constraint \eqref{eq:st_rank_W} relaxes problem (P2.1) into
\begin{subequations}
  \begin{align}\notag
   \text{(P2.2)}:\maximize_{\mathbf W,~\mathbf R_0,~t} &\quad  t\\ \notag
    \text { s.t. }& \quad \eqref{eq:st_PD},~\eqref{eq:st_SINR_W},~\eqref{eq:st_powe_W},~\text{and}~\eqref{eq:st_semi_W}.
  \end{align}
\end{subequations}

In the following, we adopt the SCA method to iteratively handle the non-convex constraint in \eqref{eq:st_PD}. Let $\mathbf W^{(k)}$ and $\mathbf R_0^{(k)}$ denote the local points of $\mathbf W$ and $\mathbf R_0$ at iteration $k\ge 1$ of SCA, respectively. Accordingly, we define $f_q(\mathbf W) = \mathrm{tr}^2(\mathbf W\mathbf a_q^*\mathbf a_q^T)$. Its first-order Taylor expansion at the local point $\mathbf W^{(k)}$ at iteration $k$ is given by
\setcounter{equation}{31} 
\begin{equation}\label{eq:first_order_Taylor}
\begin{split}
f_q(\mathbf W) &= \mathrm{tr}^2(\mathbf W\mathbf a_q^*\mathbf a_q^T)\\
&\ge \mathrm{tr}^2(\mathbf W^{(k)}\mathbf a_q^*\mathbf a_q^T) + 2\mathrm{tr}(\mathbf W^{(k)}\mathbf a_q^*\mathbf a_q^T)\\
&\quad\mathrm{tr}((\mathbf W-\mathbf W^{(k)})\mathbf a_q^*\mathbf a_q^T)\\
&=2\mathrm{tr}(\mathbf W^{(k)}\mathbf a_q^*\mathbf a_q^T)\mathrm{tr}(\mathbf W\mathbf a_q^*\mathbf a_q^T)\\
&\quad-\mathrm{tr}^2(\mathbf W^{(k)}\mathbf a_q^*\mathbf a_q^T)\\
&\triangleq \tilde f_q^{(k)}(\mathbf W). 
\end{split}
\end{equation}
By replacing $f_q(\mathbf W)$ in \eqref{eq:st_PD} with its lower bound surrogate function $\tilde f_q^{(k)}(\mathbf W)$ in \eqref{eq:first_order_Taylor}, a sequence of convex subproblems can be constructed and solved iteratively, yielding a locally optimal solution to problem (P2.2). Specifically, at iteration $k$, we solve
\begin{subequations}
  \begin{align}\notag
   &\text{(P2.2.$k$)}:\maximize_{\mathbf W,~\mathbf R_0,~t} \quad  t\\
    &\text {s.t.} \quad M_\mathrm{r}|\alpha_q|^4\tilde f_q^{(k)}(\mathbf W)/\sigma_\mathrm{s}^2+2|\alpha_q|^2\mathrm{tr}(\mathbf R_0 \mathbf a_q^*\mathbf a_q^T) \ge t,\forall q\in \mathcal Q,\\\notag 
    &\quad \quad \eqref{eq:st_SINR_W},~\eqref{eq:st_powe_W},~\text{and}~\eqref{eq:st_semi_W}.
  \end{align}
\end{subequations}
Since problem (P2.2.$k$) is convex, its global optimum can be efficiently obtained by exploiting standard convex optimization solvers such as CVX \cite{cvx}. Subsequently, the optimal solution to problem (P2.2.$k$), i.e, $\{\tilde{\mathbf W}^{(k)},\tilde{\mathbf R}_0^{(k)}\}$ , is adopted as the local point for iteration $k+1$. Note that each iteration results in a non-decreasing objective value of problem (P2.2) as 
\begin{equation}
\begin{split}
&M_\mathrm{r}|\alpha_q|^4 f_q(\mathbf W^{(k+1)})/\sigma_\mathrm{s}^2+2|\alpha_q|^2\mathrm{tr}(\mathbf R_0^{(k+1)} \mathbf a_q^*\mathbf a_q^T)\\
\ge~&M_\mathrm{r}|\alpha_q|^4 \tilde f_q^{(k)}(\mathbf W^{(k+1)})/\sigma_\mathrm{s}^2+2|\alpha_q|^2\mathrm{tr}(\mathbf R_0^{(k+1)} \mathbf a_q^*\mathbf a_q^T)\\
\ge~&M_\mathrm{r}|\alpha_q|^4\tilde f_q^{(k)}(\mathbf W^{(k)})/\sigma_\mathrm{s}^2+2|\alpha_q|^2\mathrm{tr}(\mathbf R_0^{(k)} \mathbf a_q^*\mathbf a_q^T)\\
=~&M_\mathrm{r}|\alpha_q|^4f_q(\mathbf W^{(k)})/\sigma_\mathrm{s}^2+2|\alpha_q|^2\mathrm{tr}(\mathbf R_0^{(k)} \mathbf a_q^*\mathbf a_q^T),\forall q\in \mathcal Q.
\end{split}
\end{equation}
Therefore, the convergence of the proposed iterative SCA-based algorithm is guaranteed. By iteratively solving problem (P2.2.$k$), a local optimal solution of problem (P2.2) is obtained. 

Finally, we reconstruct a rank-one solution to problem (P2.1). Without loss of generality, let $\tilde{\mathbf W}$ and $\tilde{\mathbf R}_0$ denote the optimized solution to problem (P2.2). If $\tilde{\mathbf W}$ satisfies the rank-one constraint \eqref{eq:st_rank_W}, the optimized beamforming vector $\mathbf w$ for  problem (P2) can be directly recovered via eigenvalue decomposition (EVD) of $\tilde{\mathbf W}$. Otherwise, we reconstruct a feasible solution as \cite{9124713,9652071},
\begin{subequations}
\begin{align}
\mathbf w &= (\mathbf h^H \tilde{\mathbf W} \mathbf h)^{-1/2}\tilde{\mathbf W} \mathbf h,\\
\mathbf R_0 &= \tilde{\mathbf R}_0 + \tilde{\mathbf W} - \mathbf w\mathbf w^H.
\end{align}
\end{subequations}
This construction preserves the effective communication signal power toward the authorized UAV, i.e., $|\mathbf w^H\mathbf h|^2= \mathbf h\tilde{\mathbf W}\mathbf h^H$, and therefore guarantees satisfaction of the SINR constraint in \eqref{eq:st_SINR} \cite{9124713,9652071}. The detailed transmit beamforming design algorithm is summarized in Algorithm~\ref{algorithm:P2}.

To gain furthermore insight, we discuss a special case with $Q=1$, which corresponds to detecting the presence of a target at a known location.
\begin{remark}
When $Q=1$, the optimality of problem (P2.2.$k$) is attained with $\mathrm{rank}(\tilde{\mathbf W})=1$. The detailed proof is given as follows. First, we assume that the optimal solution to (P2.2.$k$) satisfies $\mathrm{rank}(\tilde{\mathbf W})>1$. Then, given an optimal covariance matrix $\tilde{\mathbf R}_0$, problem (P2.2.$k$) reduces to a semi-definite program (SDP) with $m=2$ affine constraints on $\mathbf W$, denoted as problem (P2.2.$k'$). The rank of the optimal solution to problem (P2.2.$k'$), denoted by $r= \mathrm{rank}(\tilde{\mathbf W}')$, satisfies $r(r+1)/2\le m$ \cite{5447068,pataki1998rank}. Thus, the optimal solution to problem (P2.2.$k'$) and that to problem (P2.2.$k$)  satisfies $\mathrm{rank}(\tilde{\mathbf W})=\mathrm{rank}(\tilde{\mathbf W}')\le 1$. This contradicts the initial assumption. Therefore, the solution to problem (P2.2.$k$) must satisfy the rank-one constraint in \eqref{eq:st_rank_W}. Finally, the optimized beamforming vector $\mathbf w$ to problem (P2) can be directly recovered via EVD of $\tilde{\mathbf W}$, without requiring additional rank reduction procedures.
\end{remark}

\begin{algorithm}[t]
\caption{Proposed Algorithm for Solving Problem (P2)}
\label{algorithm:P2}
\begin{algorithmic}[1] 
\STATE Initialize the iteration index $k= 1$ and set  $\mathbf W^{(1)}$ with $\mathbf W^{(1)} \succeq \mathbf 0$ and $\mathrm{rank}(\mathbf W^{(1)})=1$.
  \REPEAT
  \STATE Construct the first-order Taylor expansion $\tilde{f}^{(k)}(\mathbf W^{(k)})$ based on the local point $\mathbf W^{(k)}$.
  \STATE Solve problem (P2.2.$k$) via CVX and denote the optimal solution as $\tilde{\mathbf W}^{(k)}$ and $\tilde{\mathbf R}_0^{(k)}$.\\
  \STATE Update the local point as $\mathbf W^{(k+1)}\gets\tilde{\mathbf W}^{(k)}$.
   \STATE Update $k\gets k+1$.
  \UNTIL either the convergence condition holds or the iteration limit is reached.
  \STATE Reconstruct the rank-one solution to (P2).
\end{algorithmic}
\end{algorithm}

\section{Transmit Beamforming Design With Only Gaussian Information Signals}\label{sec:BF_design_Gaussina}
This section presents the transmit beamforming design for the case where the BS relies solely on Gaussian-information‑bearing signals for both sensing and communication functionalities. This scenario corresponds to directly leveraging the existing communication signals for ISAC.

In this case, maximizing the detection probability is equivalent to maximizing $\gamma_{\mathrm{c},q}$ as given in \eqref{eq:PD_FA_Gaussian}. Thus, the detection probability maximization problem becomes
\begin{subequations}
  \begin{align}\notag
   \text{(P3)}:\maximize_{\mathbf w} \quad \min_{q\in \mathcal Q} &\quad  \{|\alpha_q|^2|\mathbf a_q^T\mathbf w|^2\}\\ \label{eq:st_SINR_2}
    \text { s.t. }& \quad \frac{|\mathbf h^H \mathbf w|^2}{\sigma_\mathrm{c}^2}\ge \gamma_0,\\\label{eq:st_powe_2}
    &\quad \|\mathbf w\|^2 \le P.
  \end{align}
\end{subequations}

Problem (P3) is a quadratically constrained quadratic program (QCQP), which can be efficiently addressed via the SDR method \cite{5447068}. By defining $\mathbf W = \mathbf w \mathbf w^H$, $\mathbf W \succeq \mathbf 0$, and $\mathrm{rank}(\mathbf W) = 1$ and introducing an auxiliary optimization variable $u$, problem (P3) is reformulated as 
\begin{subequations}
  \begin{align}\notag
   \text{(P3.1)}:\maximize_{\mathbf W,~u} &\quad  u \\
    \text {s.t.}& \quad |\alpha_q|^2\mathrm{tr}(\mathbf W\mathbf a_q^*\mathbf a_q^T)\ge u,\forall q\in \mathcal Q,\\
    & \quad \mathrm{tr}(\mathbf W\mathbf h \mathbf h^H)\ge \gamma_0\sigma_\mathrm{c}^2,\\
    &\quad \mathrm{tr}(\mathbf W) \le P,\\
    &\quad \mathbf W \succeq \mathbf 0,\\ \label{st:rank_one_P3}
    &\quad \mathrm{rank}(\mathbf W) = 1,
  \end{align}
\end{subequations}
which is non-convex due to the rank-one constraint in \eqref{st:rank_one_P3}.
Next, we relax the rank-one constraint in \eqref{st:rank_one_P3} and denote the resulting problem as (P3.2). Problem (P3.2) is convex, and its optimal solution $\mathbf W^\star$ can be efficiently obtained using the convex optimization solver CVX. If $\mathbf W^\star$ satisfies the rank-one constraint, the optimal beamforming vector is directly recovered by performing the EVD of $\mathbf W^\star$. Otherwise, the Gaussian randomization technique is employed to construct a feasible rank-one solution to problem (P3) \cite{5447068}. In particular, we first generate a set of samples according to $\mathbf v \sim \mathcal {CN}(\mathbf 0, \mathbf W^\star)$. To satisfy the communication SNR requirement in \eqref{eq:st_SINR_2}, each sample is rescaled as $\hat{\mathbf w} = \sqrt{\frac{\sigma_\mathrm{c}^2\gamma_0}{|\mathbf h^H \mathbf v|^2}}\mathbf v$.
Subsequently, we verify whether the resulting vector $\hat{\mathbf w}$ satisfies the transmit power constraint given in \eqref{eq:st_powe_2}. Finally, the optimized beamforming vector $\mathbf w$ is chosen as the candidate that yields the largest minimum objective value in problem (P3) among all feasible realizations.

In addition to the general communication channel models, it is noteworthy that when the communication channel $\mathbf h$ is LoS-dominated, the SDR of problem (P3.1) becomes tight. Specifically, there always exists a globally optimal solution $\mathbf W^\star$ to problem (P3.2) such that $\mathrm{rank}(\mathbf W^\star)=1$. In this case, the Gaussian randomization procedure is not needed. Detailed proofs can be found in \cite{10086626,4305444}.

Next, we discuss a special case with $Q=1$ and provide the corresponding closed-form beamforming design and ISAC performance characterization.
\begin{remark}
When $Q=1$, the optimal solution of (P3) can be obtained by solving the Karush-Kuhn-Tucker (KKT) conditions \cite{9652071,11391525}, i.e.,
\begin{equation}
\mathbf w^\star = \begin{cases} \sqrt{P}\frac{\mathbf a_1^*}{\|\mathbf a_1\|}, &\text{if}~\gamma_0\le\frac{P|\mathbf h^H \mathbf a_1^*|^2}{M_\mathrm{t}\sigma_\mathrm{c}^2},\\
x_1 \mathbf e_1 + x_2 \mathbf e_2, \quad &\text{otherwise},
\end{cases}
\end{equation}
where $\mathbf e_1 = \frac{\mathbf h}{\|\mathbf h\|}$, $\mathbf e_2 = \frac{\mathbf a_1^*-\mathbf e_1^H\mathbf a_1^*\mathbf e_1}{\|\mathbf a_1^*-\mathbf e_1^H\mathbf a_1^*\mathbf e_1\|}$, $x_1= \sqrt{\frac{\gamma_0\sigma_\mathrm{c}^2}{\|\mathbf h\|^2}}\frac{\mathbf e_1^H\mathbf a_1^*}{|\mathbf e_1^H\mathbf a_1^*|}$, and $x_2 = \sqrt{P-\frac{\gamma_0\sigma_\mathrm{c}^2}{\|\mathbf h\|^2}} \frac{\mathbf e_2^H\mathbf a_1^*}{|\mathbf e_2^H\mathbf a_1^*|}$. 
It can be observed that when the communication SNR requirement is smaller than the threshold achieved by allocating all transmit signal to the sensing direction, the optimal beamforming solution yields a sensing‑directed beampattern. Conversely, once the communication SNR requirement exceeds that threshold, the optimal beamforming vector lies in the span of the communication channel vector $\mathbf h$ and the steering vector $\mathbf a^*$. Based on the above optimal solution, the maximum detection probability under the SNR constraint $\gamma_0$ is 
\begin{equation}
 P'_{\text{D},1} = \mathcal{Q}_{\chi^2_{2L}\left(0\right)}\left(\frac{\mathcal{Q}^{-1}_{\chi^2_{2L}\left(0\right)}\left(P_{\text{FA},1}\right)}{1+|\alpha_1|^2M_\mathrm{r}/\sigma_\mathrm{s}^2\mathcal P^\star}\right),
\end{equation}
where
\begin{equation}
\begin{split}
\mathcal P^\star
=& \begin{cases} PM_\mathrm{t}, \qquad \qquad \qquad \qquad \qquad \quad  \text{if}~\gamma_0\le\frac{P|\mathbf h^H \mathbf a_1^*|^2}{M_\mathrm{t}\sigma_\mathrm{c}^2},\\
\left(\frac{\sqrt{\gamma_0\sigma_\mathrm{c}^2}|\mathbf h^H\mathbf a_1^*|}{\|\mathbf h\|^2}\right.\\
\left.+\sqrt{\left(P-\frac{\gamma_0\sigma_\mathrm{c}^2}{\|\mathbf h\|^2}\right)\left(M_\mathrm{t}-\frac{|\mathbf h^H\mathbf a_1^*|^2}{\|\mathbf h\|^2}\right)}\right)^2, \quad \text{otherwise}.
\end{cases}
\end{split}
\end{equation}
\end{remark}
When the communication SNR requirement is sufficiently low, i.e, $\gamma_0\le\frac{P|\mathbf h^H \mathbf a_1^*|^2}{M_\mathrm{t}\sigma_\mathrm{c}^2}$, the detection probability is constant:
$
 P'_{\text{D},1} = \mathcal{Q}_{\chi^2_{2L}\left(0\right)}\left(\frac{\mathcal{Q}^{-1}_{\chi^2_{2L}\left(0\right)}\left(P_{\text{FA},1}\right)}{1+P|\alpha_1|^2M_\mathrm{r}M_\mathrm{t}/\sigma_\mathrm{s}^2}\right)
$.
Otherwise, when the communication SNR requirement becomes larger, the detection probability gradually decreases to the value achieved by the communication‑only beampattern, i.e., 
$
 P'_{\text{D},1} = \mathcal{Q}_{\chi^2_{2L}\left(0\right)}\left(\frac{\mathcal{Q}^{-1}_{\chi^2_{2L}\left(0\right)}\left(P_{\text{FA},1}\right)}{1+P|\alpha_1|^2M_\mathrm{r}|\mathbf h^H\mathbf a_1^*|^2/(\|\mathbf h\|^2\sigma_\mathrm{s}^2)}\right)
$.

\section{Simulation Results}\label{sec:numerical_results}
This section presents simulation results to evaluate the detection performance of the proposed NP-based detector and to characterize the ISAC performance boundary achieved by transmit beamforming optimization. The communication channel $\mathbf h$ is modeled as Rician fading with 
$\mathbf h = \sqrt{\frac{K}{K+1}}\mathbf h_\mathrm{los} + \sqrt{\frac{1}{K+1}}\mathbf h_\mathrm{nlos}$,
where $K=1$ is the Rician factor, and $\mathbf h_\mathrm{los}$ and $\mathbf h_\mathrm{nlos}$ are the LoS and Rayleigh fading components, respectively\cite{11391525}. The path-loss model is 
$L(d) = L_0\left(\frac{d}{d_0}\right)^{-\beta_0}$,
where $L_0=-30$~dB is the path-loss at distance $d_0 = 1$~m, $d$ is the transmission distance, and $\beta_0=2.5$ is the path-loss exponent, respectively\cite{11391525}. The waveform carrier frequency is set to $f = 800$~MHz\cite{11087656}. The target’s RCS is set to $\sigma_\mathrm{t}=0.5~\text{m}^2$. The distances between the BS and the sensing region and that between the sensing region and the sensing receiver are set to $d_{1,q}=d_{2,q}=300$~m, $\forall q\in \mathcal{Q}$, respectively. The desired sensing angles with respect to the BS are set to $[-\pi/80, \pi/80]$, with the number of samples being $Q = 50$. The remaining system parameters are set to $M_\mathrm{t} =M_\mathrm{r} =16$, $P = 30$~dBm, $\sigma_\mathrm{c}^2=\sigma_\mathrm{s}^2=-80$~dBm, $d_\mathrm{a}=\lambda/2$, $L=1024$, and $P_{\text{FA},q}=10^{-3},\forall q\in\mathcal Q$, respectively\cite{11391525}. 

\subsection{Sensing Performance Analysis}

In this subsection, we consider the sensing-only scenario and analyze the detection probability at a sampling point as derived in Section~\ref{sec:detection_probability_derivation} under various system configurations. 

\begin{figure}[t] 
        \centering
        \includegraphics[width=0.36\textwidth]{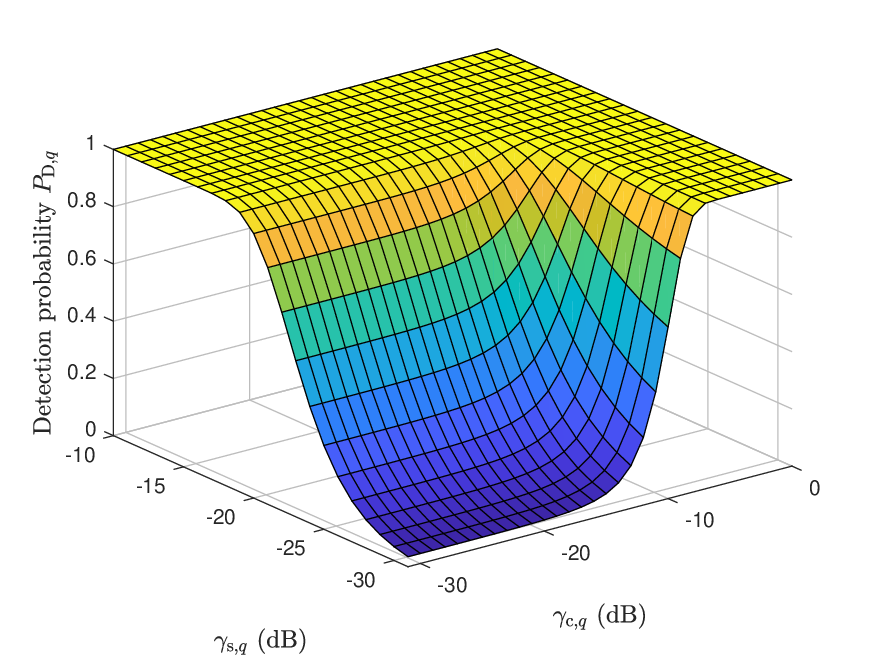}
        \vspace{-5pt}\caption{Detection probability versus the ratios of received deterministic/Gaussian signals powers to the noise power, where $L=1024$.}
        \label{fig:PD_gamma_c_gamma_s}\vspace{-10pt}
\end{figure}

\begin{figure}[t] 
        \centering
        \includegraphics[width=0.36\textwidth]{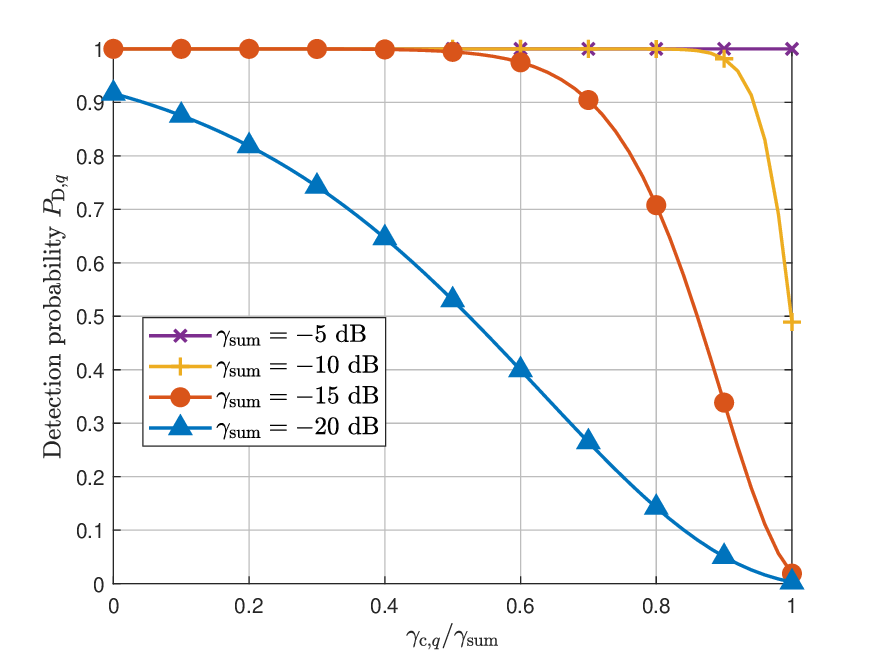}
        \vspace{-5pt}\caption{Detection probability versus the ratio of received deterministic sensing SNR to the sum SNR, where $L=1024$.}
        \label{fig:PD_gamma}\vspace{-10pt}
\end{figure}

First, we evaluate the sensing performance when the BS transmits a combination of Gaussian-information‑bearing and deterministic sensing signals. Fig.~\ref{fig:PD_gamma_c_gamma_s} illustrates the target detection probability $P_{\mathrm{D},q}$ versus the sensing SNRs $\gamma_{\mathrm{c},q}$ and $\gamma_{\mathrm{s},q}$. It can be observed that increasing either the deterministic or Gaussian signal power enhances the detection probability, since both signal components contribute useful target-related information to the received echo signals. 
Meanwhile, within a given region, increasing the sensing SNR $\gamma_{\mathrm{c},q}$ or $\gamma_{\mathrm{s},q}$ leads to a sharp increase in the detection probability. This indicates that even a marginal improvement in sensing SNR can result in a substantial gain in detection probability. The underlying reason is that the detection probability follows the integral form of the right-tail probability of a chi‑squared distribution, whose tail decays exponentially.
Furthermore, when either $\gamma_{\mathrm{c},q}$ or $\gamma_{\mathrm{s},q}$ is sufficiently small (e.g., $-30$ dB), the minimum $\gamma_{\mathrm{s},q}$ (approximately $-15~\text{dB}$) required to achieve a detection probability of one is smaller than the corresponding $\gamma_{\mathrm{c},q}$ threshold (approximately $-5~\text{dB}$). This observation indicates that deterministic sensing signals outperform Gaussian-information-bearing signals for sensing as they enable sequence-based correlation at the receiver, in contrast to the energy-based comparison with random signals.

In practical sensing systems, the transmitter's power constraint limits the feasible operating region of the sensing SNRs, $\gamma_{\mathrm{c},q}$ and $\gamma_{\mathrm{s},q}$. Consequently, we analyze the detection probability under a fixed total SNR budget, i.e., $\gamma_{\mathrm{c},q} + \gamma_{\mathrm{s},q} = \gamma_\mathrm{sum}$. Fig.~\ref{fig:PD_gamma} shows the detection probability versus the ratio $\gamma_{\mathrm{c},q} / \gamma_\mathrm{sum}$. When $\gamma_\mathrm{sum}$ is low, e.g., $-20$~dB, allocating more power to the Gaussian-information-bearing component $\gamma_{\mathrm{c},q}$ decreases the detection probability. In this regime, exploiting superimposed signals yields a detection probability that is superior to that of exploiting only Gaussian-information‑bearing signals but inferior to that attained with only deterministic sensing signals. This result confirms that the inherent randomness of Gaussian-information-bearing signals degrades detection performance. Conversely, when $\gamma_\mathrm{sum}$ is sufficiently large ($-5$~dB), the detection probabilities under various values of $\gamma_{\mathrm{c},q}$ all equal one. This result demonstrates the feasibility of leveraging existing communication signals and system architecture for joint target detection when the  sensing SNR is sufficiently large.

\begin{figure}[t] 
        \centering
        \includegraphics[width=0.36\textwidth]{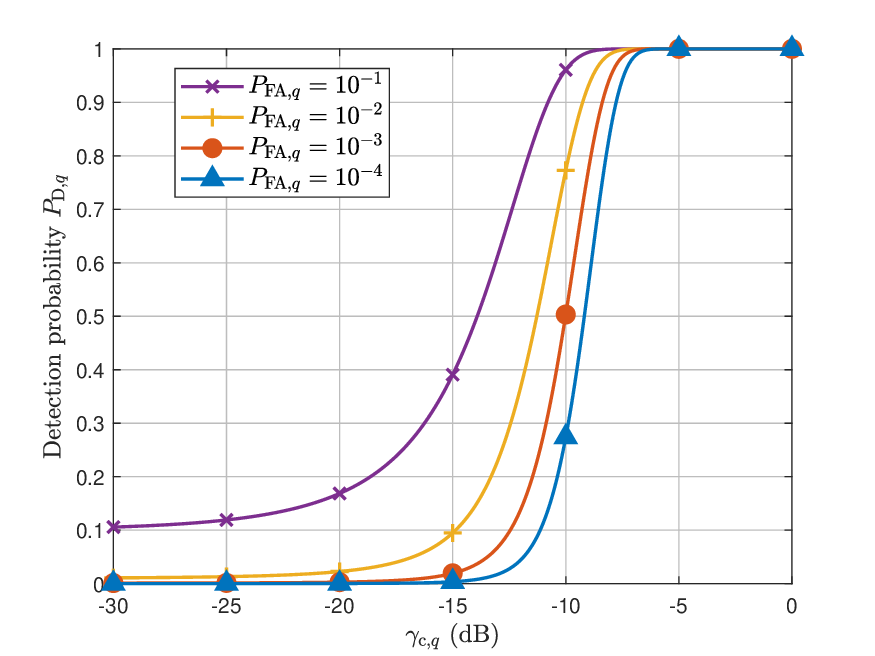}
        \vspace{-5pt}\caption{Detection probability versus the received sensing SNR $\gamma_{\mathrm{c},q}$ under various false-alarm probabilities, where $L=1024$.}
        \label{fig:PD_gamma_FA}\vspace{-10pt}
\end{figure}

Next, we investigate the detection performance when the BS transmits only Gaussian-information-bearing signals. Fig.~\ref{fig:PD_gamma_FA} shows the detection probability $P_{\mathrm{D},q}$ versus the received sensing SNR $\gamma_{\mathrm{c},q}$ under various false-alarm probabilities. As expected from Remark~\ref{re:P_D_random}, the detection probability increases monotonically with the sensing SNR and exhibits a cliff-like increase as the sensing SNR increases from $-15$~dB to $-5$~dB, owing to the integral form of the detection probability. Furthermore, as the false-alarm probability decreases from $10^{-1}$ to $10^{-4}$, the minimum SNR required for the detection probability to approach one increases from $-7.5$~dB to $2.5$~dB. This result is due to the fact that a stricter false-alarm constraint necessitates a higher decision threshold to reduce the likelihood of declaring a target under the null hypothesis. Consequently, a higher sensing SNR is required to ensure that the test statistic exceeds the threshold. These results delineate the SNR region required for reliable sensing in practical systems.

\begin{figure}[t] 
        \centering
        \includegraphics[width=0.36\textwidth]{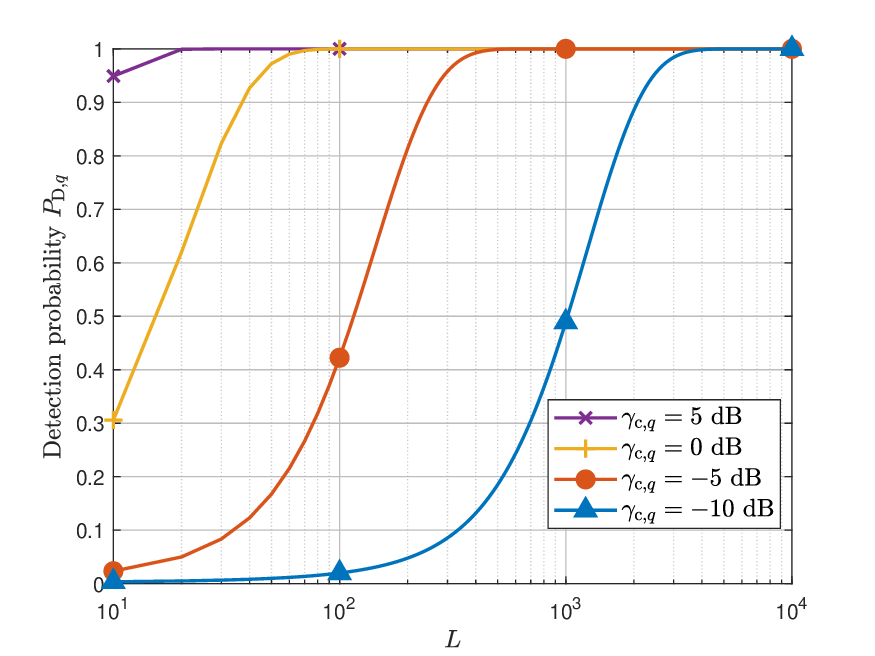}
        \vspace{-5pt}\caption{Detection probability versus the sensing duration time $L$ under various sensing SNRs, where $P_{\text{FA},q}=10^{-3}$.}
        \label{fig:PD_gamma_L}\vspace{-10pt}
\end{figure}
Fig.~\ref{fig:PD_gamma_L} shows the detection probability versus the sensing duration $L$ for various sensing SNRs at a fixed false-alarm probability $P_{\text{FA},q}=10^{-3}$. Increasing the sensing duration significantly improves the detection probability. This is because the energy-based detector in \eqref{eq:detector_Gaussian} accumulates the received signal energy over the entire sensing interval. Consequently, a longer sensing duration amplifies the distinction of test statistic between these two hypotheses and enables more reliable detection. This demonstrates that an extended sensing duration can compensate for a lower sensing SNR. In particular, to achieve a detection probability of one, the minimum required number of sensing symbols are approximately $30$, $90$, $600$, and $4500$ when the sensing SNR is $5$~dB, $0$~dB, $-5$~dB, and $-10$~dB, respectively. This result illustrates a fundamental trade-off between sensing time and sensing SNR under a fixed detection probability requirement.

\begin{figure}[t] 
        \centering
        \includegraphics[width=0.36\textwidth]{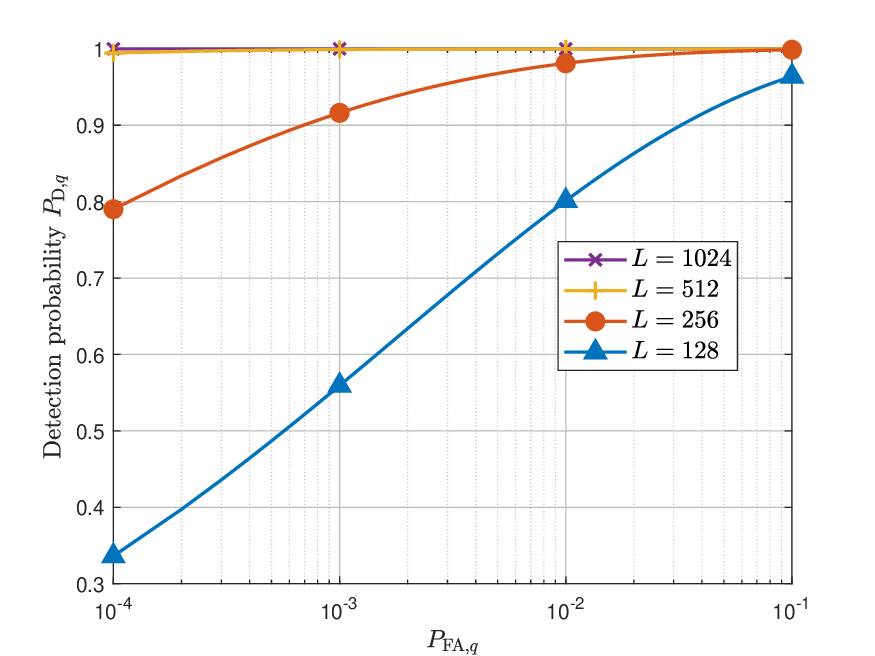}
        \caption{Detection probability versus the sensing duration time $L$ under various false-alarm probabilities, where $\gamma_{\mathrm{c},q}=-5$~dB.}
        \label{fig:PD_L_FA}\vspace{-10pt}
\end{figure}
Fig.~\ref{fig:PD_L_FA} illustrates the relationship between detection probability and false-alarm constraints under different sensing duration $L$ when $\gamma_{\mathrm{c},q} = -5$ dB. The results exhibit the classic detection trade-off: detection probability increases with the permissible false-alarm probability. This trade-off is pronounced for short sensing durations (e.g., $L = 128, 256$), as the likelihood functions under these two hypotheses are broad and overlap significantly. However, for longer durations (e.g., $L = 512, 1024$), increases in false-alarm probability yield only marginal gains in detection performance, because the likelihood functions become concentrated according to the central limit theorem.

\subsection{ISAC Performance Analysis}
Next, we explore the trade-off between detection probability and communication-rate in the ISAC scenario. The proposed transmit beamforming design is compared with the following benchmark schemes.

\subsubsection{Target Detection Utilizing only Deterministic Sensing Signal}
The BS simultaneously transmits both deterministic sensing and Gaussian-information-bearing signals. The sensing receiver only leverages the deterministic sensing signal for target detection and treats the Gaussian-information-bearing signal as interference. The traditional matched-filter is applied to detect the presence of the target. In this case, the transmit beamforming at the BS is optimized to maximize the minimum power of deterministic signals toward the sensing region, subject to the minimum communication SINR requirement and the maximum transmit power budget.

\subsubsection{Time Switching} This scheme employs a time-division strategy, where the BS transmits deterministic signals for sensing and Gaussian signals for communication in separate phases. The sensing receiver performs target detection in the sensing phase. By adjusting the time allocation $L_\mathrm{s}\ge0$ and $L_\mathrm{c}\ge0$ to sensing and communication, respectively,  with $L_\mathrm{s}+L_\mathrm{c}=L$, the average communication rate $R' = (L_\mathrm{c}/L)R$ is chosen to satisfy the minimum rate requirement.

\subsubsection{Beampattern Gain Maximization \cite{10086626}}
The BS transmits both types of signals and the transmit beamforming design at the BS is optimized to maximize the minimum beampattern gain across the sensing region, under the constraints of the authorized UAV's SINR requirement and the BS's maximum transmit power. Accordingly, the proposed detector is utilized to exploit both signal types for sensing performance enhancement.

\begin{figure}[t]
        \centering
        \includegraphics[width=0.36\textwidth]{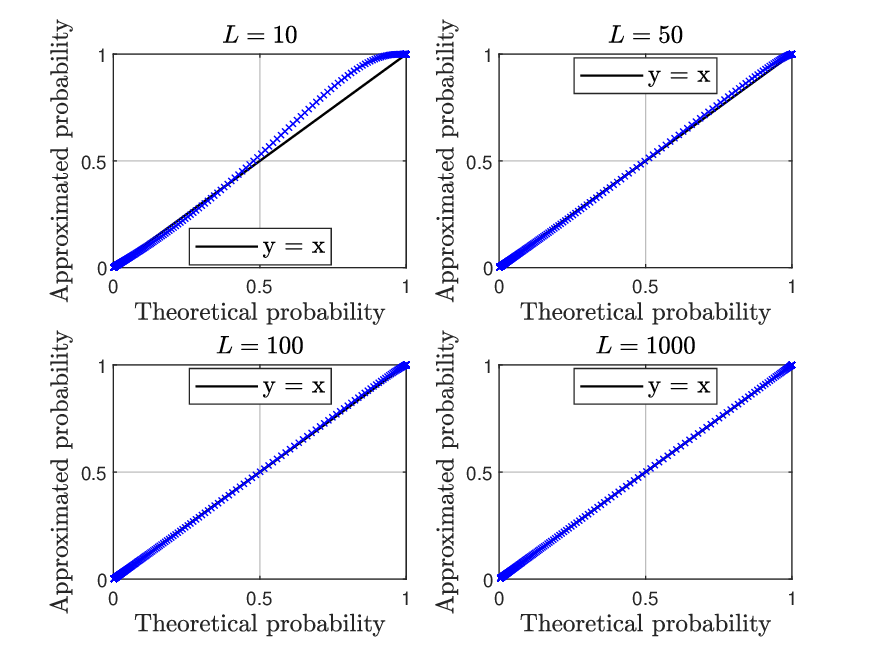}
        \vspace{-5pt}\caption{Quantile-quantile plot comparing the approximated and theoretical detection probabilities over the SNR range $\gamma_{\mathrm{c},q}, \gamma_{\mathrm{s},q} \in [-40, 10]~\text{dB}$.}
        \label{fig:approximation}\vspace{-10pt}
\end{figure}
\begin{figure}[t]
        \centering
        \includegraphics[width=0.36\textwidth]{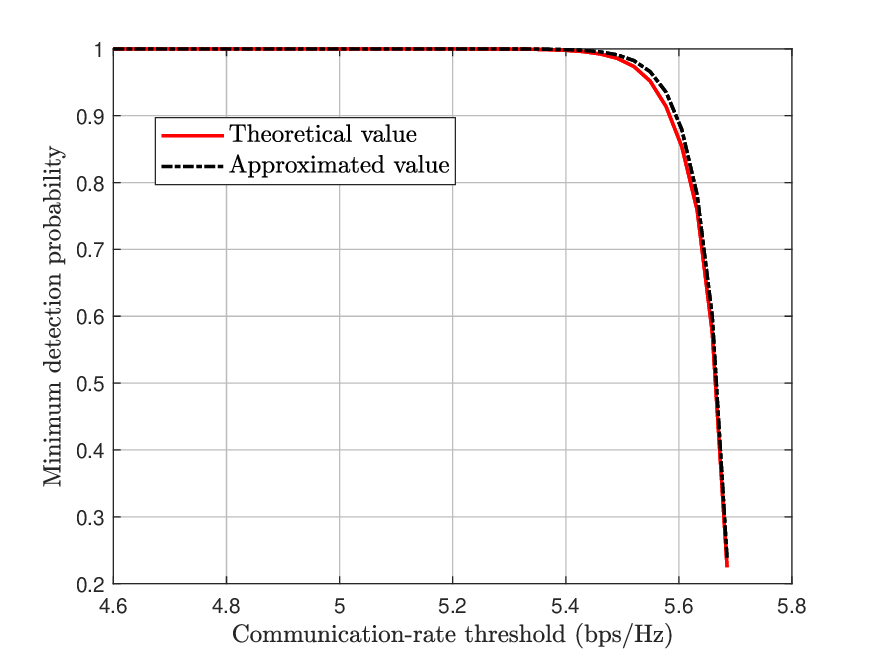}
        \vspace{-5pt}\caption{Comparison between minimum theoretical and approximated target detection probabilities over the sensing region $\mathcal Q$ versus the achievable communication-rate threshold.}
        \label{fig:PD_comparison}\vspace{-10pt}        
\end{figure}

Fig.~\ref{fig:approximation} validates the accuracy of the approximations in Proposition~\ref{prop:approximation_L} under various values of the sensing duration $L$. For $L=10$ and $L=50$, the approximated quantiles exhibit significant and slight bias relative to their theoretical counterparts, respectively, particularly within the upper interval $[0.5,1]$. In contrast, for $L=100$ and $L=1000$, the quantile-quantile points adhere closely to the $y=x$ reference line across all tested values of $\gamma_{\mathrm{c},q}$ and $\gamma_{\mathrm{s},q}$. These observations confirm that the proposed approximations provide high accuracy when the ISAC duration $L$ is sufficiently large.

Fig.~\ref{fig:PD_comparison} illustrates the minimum theoretical detection probability over the sensing region $\mathcal Q$ alongside its approximation under the proposed beamforming design. It is shown that when the proposed approximation is employed as the objective for transmit beamforming optimization, the resulting exact detection probability closely approaches the approximated curve. This result further verifies the approximation’s accuracy and demonstrates its suitability for practical ISAC design.

\begin{figure}[t]
        \centering
        \includegraphics[width=0.36\textwidth]{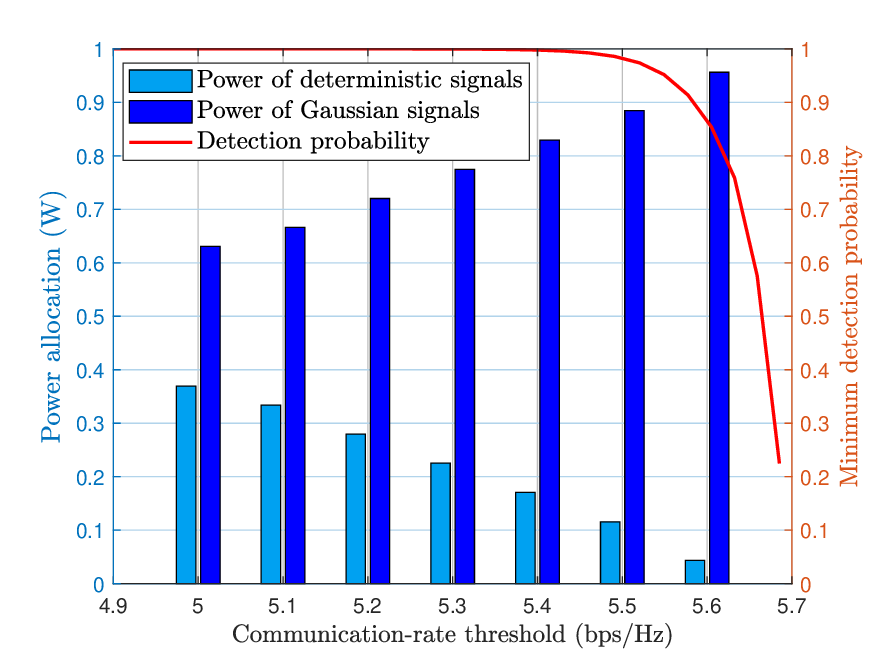}
        \vspace{-5pt}\caption{Power allocations between deterministic and Gaussian signals, and the corresponding minimum detection probability of the proposed design versus the achievable communication-rate threshold.}
        \label{fig:power_allocation}\vspace{-10pt}
\end{figure}
\begin{figure}[t]
        \centering
        \includegraphics[width=0.36\textwidth]{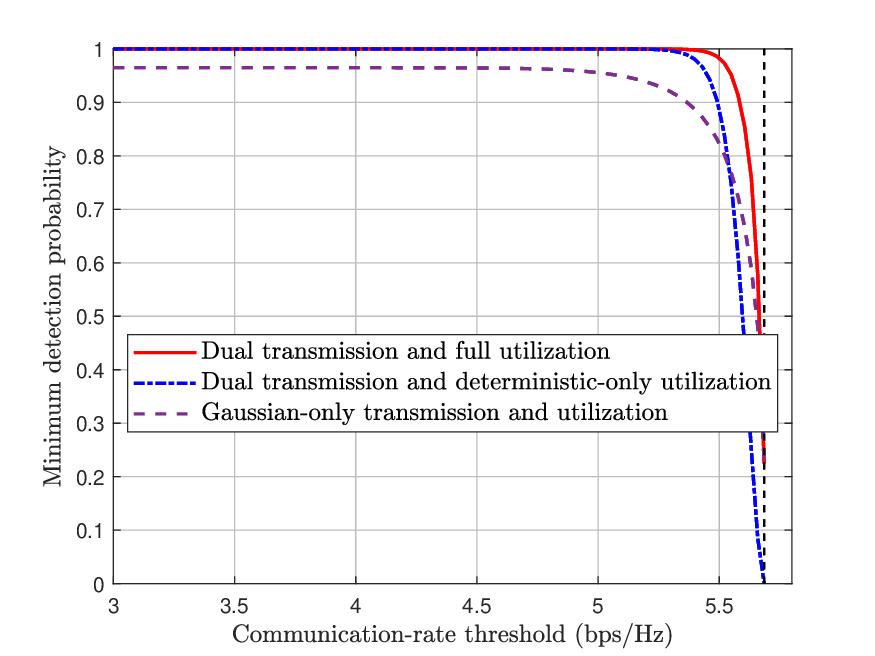}       
        \vspace{-5pt}\caption{Comparison of minimum detection probability over the sensing region $\mathcal{Q}$ for different signal models and detectors.}
        \label{fig:PD_detector}\vspace{-10pt}        
\end{figure}
\begin{figure}[t]
        \centering
        \includegraphics[width=0.36\textwidth]{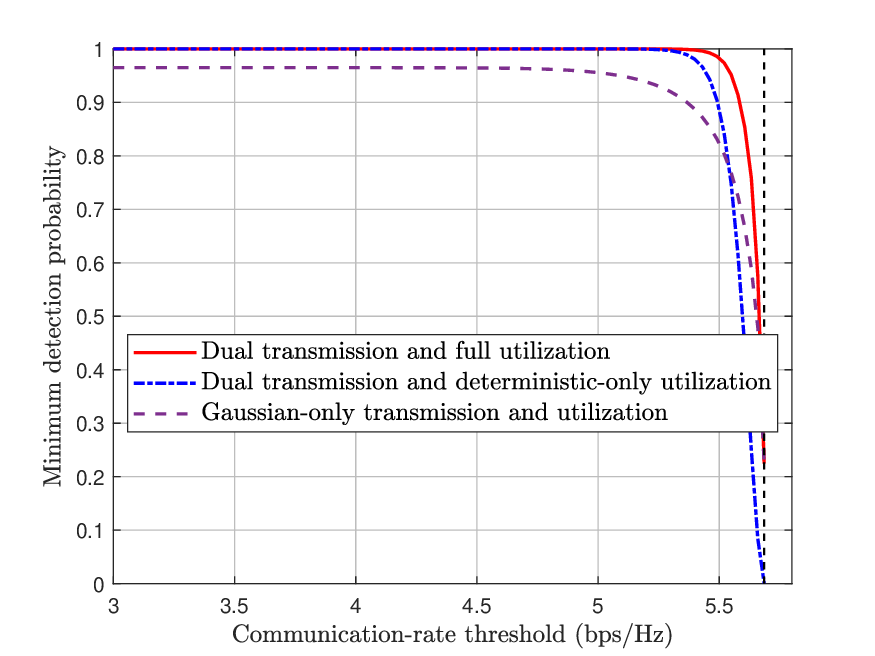}        
        \vspace{-5pt}\caption{Comparison of minimum detection probability over the sensing region $\mathcal{Q}$ for different different beamforming design methods.}
        \label{fig:PD_BF}\vspace{-10pt}        
\end{figure}

Fig.~\ref{fig:power_allocation} shows the power allocations between deterministic and random signals, as well as the trade-off between the minimum detection probability over the sensing region $\mathcal Q$ and the communication-rate threshold. It is observed that increasing the communication-rate threshold shifts the power allocation from deterministic sensing to Gaussian-information-bearing signals, which strengthens the received information signals and suppresses interference from the sensing components. Meanwhile, as the communication-rate requirement becomes stringent, the transmit beamforming design directs more power toward the authorized UAV. These two effects increase the signal randomness and reduce the power of the signals steered to the sensing target, collectively degrading the detection probability. These results demonstrate the performance trade-off between achievable communication rate and detection probability in the ISAC system.

Fig.~\ref{fig:PD_detector} illustrates the trade-off between the minimum target detection probability over the sensing region $\mathcal Q$ and achievable communication rate under various signal models and detectors. It is observed that the detection probabilities of all schemes decrease as the minimum communication rate requirement increases. The proposed detector that employs superimposed deterministic and Gaussian signals, consistently outperforms the other benchmark schemes. At a low communication rate threshold, most transmit power is allocated to deterministic sensing signals, allowing phase-coherent accumulation of target echo signals, thus achieving superior detection probability at these methods utilizing deterministic signals. As the communication rate requirement increases, more power is allocated to Gaussian-information-bearing signals. Consequently, these two proposed detectors that utilizing Gaussian signals achieves a higher detection probability than the benchmark relying on only deterministic signals.

Fig.~\ref{fig:PD_BF} compares the target detection probabilities achieved by different beamforming design strategies. The proposed design consistently outperforms the time-switching and beampattern gain maximization benchmark schemes, with the advantage becoming more pronounced under stringent communication rate requirements. This confirms the necessity of dedicated beamforming design when exploiting both types of signals for detection, as deterministic and Gaussian signals contribute differently to detection performance, as shown in Fig.~\ref{fig:PD_gamma_c_gamma_s}.

\section{Conclusion}\label{sec:conclusion}

This paper investigated the ISAC performance of a system involving authorized-UAV communication and bistatic unauthorized-UAV detection, where a BS simultaneously transmits a superimposed waveform consisting of deterministic sensing and Gaussian-information-bearing signals. To fully exploit these superimposed signals, we introduced an NP-based detector that jointly leverages both signal components and derived the corresponding detection probability. Furthermore, by formulating a beamforming optimization problem, we characterized the fundamental ISAC performance boundary. Simulation results demonstrated that both types of signals contribute to target detection. The non-trivial sensing-communication trade-off highlights the importance of developing adaptive waveform and beamforming design strategies for practical ISAC systems.
\appendices

\section{Proof of the Approximation in Proposition~\ref{prop:approximation_L}} \label{app:proof_of_the_approximation}
First, for a sufficiently large sensing duration $L$, the central limit theorem implies that $T'(\tilde{\mathbf y};\mathcal{H}_0)$ and $T'(\tilde{\mathbf y};\mathcal{H}_1)$ asymptotically follow Gaussian distributions, i.e., $T'(\tilde{\mathbf y};\mathcal{H}_0)\sim \mathcal{N}(\nu_1+\lambda_1,2\nu_1+4\lambda_1)$ and $T'(\tilde{\mathbf y};\mathcal{H}_1)\sim \mathcal{N}(\nu_2+\lambda_2,2\nu_2+4\lambda_2)$. In this case, the false-alarm probability is approximated as
\begin{equation}\label{eq:P_FA_delta_appro}
\begin{split}
 P_{\text{FA},q} &\approx  \bar{P}_{\text{FA},q}\\
 & = \mathrm{Pr}\left\{T'(\tilde{\mathbf y};\mathcal{H}_0)\ge \frac{2(1+\gamma_{\mathrm{c},q})\delta'}{\gamma_{\mathrm{c},q}} + \frac{2L\gamma_{\mathrm{s},q}}{\gamma_{\mathrm{c},q}^2}\right\}\\
 & = Q\left( \frac{\frac{(1+\gamma_{\mathrm{c},q})\delta'}{\gamma_{\mathrm{c},q}} - L}{\sqrt{L+\lambda_1}}\right).
 \end{split}
\end{equation}
Similarly, the detection probability is approximated as
\begin{equation}\label{eq:P_D_delta_appro}
\begin{split}
 P_{\text{D},q} &\approx  \bar{P}_{\text{D},q}\\
 & = \mathrm{Pr}\left\{T'(\tilde{\mathbf y};\mathcal{H}_1)\ge \frac{2\delta'}{\gamma_{\mathrm{c},q}} + \frac{2L\gamma_{\mathrm{s},q}}{(1+\gamma_{\mathrm{c},q})\gamma_{\mathrm{c},q}^2}\right\}\\
 &= Q\left( \frac{\frac{(1+\gamma_{\mathrm{c},q})\delta'-L\gamma_{\mathrm{s},q}(\gamma_{\mathrm{c},q}+2)}{\gamma_{\mathrm{c},q} (1+\gamma_{\mathrm{c},q})} - L}{\sqrt{L+\lambda_2}}\right).
 \end{split}
\end{equation}
Finally, based on \eqref{eq:P_FA_delta_appro} and \eqref{eq:P_D_delta_appro}, the detection probability under a given false-alarm probability is approximated as
\begin{equation}\label{eq:approximation_I}
\begin{split}
&\bar{P}_{\text{D},q} \\
=&~  Q\left(\frac{Q^{-1}(\bar{P}_{\text{FA},q})\sqrt{\gamma_{\mathrm{c},q}^2\!+\!2\gamma_{\mathrm{s},q}}\!-\!\sqrt{L}\left(\gamma_{\mathrm{c},q}^2\!+\!\gamma_{\mathrm{s},q}(2\!+\!\gamma_{\mathrm{c},q})\right)}{(1+\gamma_{\mathrm{c},q})\sqrt{\gamma_{\mathrm{c},q}^2+2\gamma_{\mathrm{s},q}(1+\gamma_{\mathrm{c},q})}}\right).
\end{split}
\end{equation}
Thus, the approximation in $(a_{1})$ follows immediately.
Next, when $\gamma_{\mathrm{c},q} \ll 1$, the terms $1 + \gamma_{\mathrm{c},q}$ and $2 + \gamma_{\mathrm{c},q}$ can be well-approximated by $1$ and $2$, respectively. In this case, the detection probability in \eqref{eq:approximation_I} is further approximated as 
\begin{equation}\label{eq:approximation_II}
\begin{split}
\bar{P}_{\text{D},q} & \approx  Q\left(\frac{Q^{-1}(\bar{P}_{\text{FA},q})\sqrt{\gamma_{\mathrm{c},q}^2+2\gamma_{\mathrm{s},q}}-\sqrt{L}\left(\gamma_{\mathrm{c},q}^2+2\gamma_{\mathrm{s},q}\right)}{\sqrt{\gamma_{\mathrm{c},q}^2+2\gamma_{\mathrm{s},q}}}\right)\\
& = Q\left(Q^{-1}(\bar{P}_{\text{FA},q})-\sqrt{L}\sqrt{\gamma_{\mathrm{c},q}^2+2\gamma_{\mathrm{s},q}}\right)\\
& \approx Q\left(Q^{-1}(P_{\text{FA},q})-\sqrt{L}\sqrt{\gamma_{\mathrm{c},q}^2+2\gamma_{\mathrm{s},q}}\right)\\
&\triangleq \tilde P_{\mathrm{D},q}. 
\end{split}
\end{equation}
Next, it is observed that when $\gamma_{\mathrm{c},q}$ is sufficiently large or approaches unity, the numerical values of the detection probabilities $\bar{P}_{\text{D},q}$ and $\tilde P_{\mathrm{D},q}$ all approach unity. Thus, the approximation in \eqref{eq:approximation_II} remains valid for $\gamma_{\mathrm{c},q} \ge 1$ or $\gamma_{\mathrm{c},q} \rightarrow 1$. Therefore, the approximation in $(a_{2})$ is obtained, which completes the proof.
\ifCLASSOPTIONcaptionsoff
  \newpage
\fi
\bibliographystyle{IEEEtran}
\bibliography{IEEEabrv,mybibfile}

\begin{thebibliography}{10}
\providecommand{\url}[1]{#1}
\csname url@samestyle\endcsname
\providecommand{\newblock}{\relax}
\providecommand{\bibinfo}[2]{#2}
\providecommand{\BIBentrySTDinterwordspacing}{\spaceskip=0pt\relax}
\providecommand{\BIBentryALTinterwordstretchfactor}{4}
\providecommand{\BIBentryALTinterwordspacing}{\spaceskip=\fontdimen2\font plus
\BIBentryALTinterwordstretchfactor\fontdimen3\font minus
  \fontdimen4\font\relax}
\providecommand{\BIBforeignlanguage}[2]{{%
\expandafter\ifx\csname l@#1\endcsname\relax
\typeout{** WARNING: IEEEtran.bst: No hyphenation pattern has been}%
\typeout{** loaded for the language `#1'. Using the pattern for}%
\typeout{** the default language instead.}%
\else
\language=\csname l@#1\endcsname
\fi
#2}}
\providecommand{\BIBdecl}{\relax}
\BIBdecl

\bibitem{song2025detection}
X.~Song, X.~Yu, J.~Xu, and D.~W.~K. Ng, ``Detection in bistatic {ISAC} with
  deterministic sensing and {Gaussian} information signals,'' in \emph{Proc.
  IEEE Int. Conf. Commun. (ICC) 2026}, Glasgow, Scotland, UK, May 2026.

\bibitem{10955337}
Y.~Jiang, X.~Li, G.~Zhu, H.~Li, J.~Deng, K.~Han, C.~Shen, Q.~Shi, and R.~Zhang,
  ``Integrated sensing and communication for low altitude economy:
  Opportunities and challenges,'' \emph{IEEE Commun. Mag.}, vol.~63, no.~12,
  pp. 72--78, Dec. 2025.

\bibitem{11098638}
Y.~Song, Y.~Zeng, Y.~Yang, Z.~Ren, G.~Cheng, X.~Xu, J.~Xu, S.~Jin, and
  R.~Zhang, ``An overview of cellular {ISAC} for low-altitude {UAV}: New
  opportunities and challenges,'' \emph{IEEE Commun. Mag.}, vol.~63, no.~12,
  pp. 88--95, Dec. 2025.

\bibitem{11131292}
Y.~Wang, G.~Sun, Z.~Sun, J.~Wang, J.~Li, C.~Zhao, J.~Wu, S.~Liang, M.~Yin,
  P.~Wang, D.~Niyato, S.~Sun, and D.~I. Kim, ``Toward realization of
  low-altitude economy networks: Core architecture, integrated technologies,
  and future directions,'' \emph{IEEE Trans. Cognit. Commun. Netw.}, vol.~11,
  no.~5, pp. 2788--2820, Oct. 2025.

\bibitem{9696263}
S.~D. Muruganathan, X.~Lin, H.-L. Määttänen, J.~Sedin, Z.~Zou, W.~A.
  Hapsari, and S.~Yasukawa, ``An overview of {3GPP} release-15 study on
  enhanced {LTE} support for connected drones,'' \emph{IEEE Commun. Standards
  Mag.}, vol.~5, no.~4, pp. 140--146, Dec. 2021.

\bibitem{3GPP_ISAC}
\BIBentryALTinterwordspacing
{3GPP}, ``{3GPP TR 38.765 V1.0.0}: Study on integrated sensing and
  communication {(ISAC)} for {NR},'' Feb. 2026. [Online]. Available:
  \url{https://portal.3gpp.org/desktopmodules/Specifications/SpecificationDetails.aspx?specificationId=4446}
\BIBentrySTDinterwordspacing

\bibitem{9737357}
F.~Liu, Y.~Cui, C.~Masouros, J.~Xu, T.~X. Han, Y.~C. Eldar, and S.~Buzzi,
  ``Integrated sensing and communications: Toward dual-functional wireless
  networks for {6G} and beyond,'' \emph{IEEE J. Sel. Areas Commun.}, vol.~40,
  no.~6, pp. 1728--1767, Jun. 2022.

\bibitem{song2025overview}
X.~Song, Y.~Fang, F.~Wang, Z.~Ren, X.~Yu, Y.~Zhang, F.~Liu, J.~Xu, D.~W.~K. Ng,
  R.~Zhang, and S.~Cui, ``An overview on {IRS-enabled} sensing and
  communications for {6G}: architectures, fundamental limits, and joint
  beamforming designs,'' \emph{Sci. China Inf. Sci.}, vol.~68, no.~5, p.
  150301, Apr. 2025.

\bibitem{11358925}
D.~Zhang, Y.~Cui, X.~Cao, N.~Su, Y.~Gong, F.~Liu, W.~Yuan, X.~Jing,
  J.~Andrew~Zhang, J.~Xu, C.~Masouros, D.~Niyato, and M.~Di~Renzo, ``Integrated
  sensing and communications over the years: An evolution perspective,''
  \emph{IEEE Commun. Surveys Tuts.}, vol.~28, pp. 5014--5048, Mar. 2026.

\bibitem{richards2005fundamentals}
M.~A. Richards, \emph{Fundamentals of Radar Signal Processing}.\hskip 1em plus
  0.5em minus 0.4em\relax Mcgraw-hill New York, 2005, vol.~1.

\bibitem{steven1993fundamentals}
S.~M. {}Kay, \emph{Fundamentals of Statistical Signal Processing Volume II:
  Detection Theory}.\hskip 1em plus 0.5em minus 0.4em\relax PTR Prentice-Hall,
  Englewood Cliffs, NJ, 1993.

\bibitem{1703855}
I.~Bekkerman and J.~Tabrikian, ``Target detection and localization using {MIMO}
  radars and sonars,'' \emph{IEEE Trans. Signal Process.}, vol.~54, no.~10, pp.
  3873--3883, Oct. 2006.

\bibitem{goldsmith2005wireless}
A.~Goldsmith, \emph{Wireless Communications}.\hskip 1em plus 0.5em minus
  0.4em\relax Cambridge University Press, 2005.

\bibitem{10977743}
C.~Zhao, Y.~Feng, H.~Luo, F.~Gao, F.~Liu, and S.~Jin, ``Networked {ISAC-based}
  {UAV} tracking and handover toward low-altitude economy,'' \emph{IEEE Trans.
  Wireless Commun.}, vol.~24, no.~9, pp. 7670--7685, Sep. 2025.

\bibitem{11296935}
Y.~Zheng, L.~Li, W.~Lin, W.~Liang, Q.~Du, and Z.~Han, ``Optimal transport
  framework for {ISAC} in low-altitude networks: Joint resource allocation for
  cooperative communication and non-cooperative localization,'' \emph{IEEE
  Trans. Commun.}, vol.~74, pp. 1984--2000, Jan. 2026.

\bibitem{10879807}
G.~Cheng, X.~Song, Z.~Lyu, and J.~Xu, ``Networked {ISAC} for low-altitude
  economy: Coordinated transmit beamforming and {UAV} trajectory design,''
  \emph{IEEE Trans. Commun.}, vol.~73, no.~8, pp. 5832--5847, Aug. 2025.

\bibitem{11072035}
X.~Ye, Y.~Mao, X.~Yu, S.~Sun, L.~Fu, and J.~Xu, ``Integrated sensing and
  communications for low-altitude economy: A deep reinforcement learning
  approach,'' \emph{IEEE Trans. Wireless Commun.}, vol.~25, pp. 351--367, Feb.
  2026.

\bibitem{11404407}
J.~Xu, X.~Zhou, H.~Zhang, and Y.~Li, ``Deep learning-based predictive
  bidirectional beamforming in {ISAC-enabled} {UAV} networks,'' \emph{IEEE
  Trans. Wireless Commun.}, vol.~25, pp. 12\,230--12\,245, Feb. 2026.

\bibitem{10147248}
Y.~Xiong, F.~Liu, Y.~Cui, W.~Yuan, T.~X. Han, and G.~Caire, ``On the
  fundamental tradeoff of integrated sensing and communications under
  {Gaussian} channels,'' \emph{IEEE Trans. Inf. Theory}, vol.~69, no.~9, pp.
  5723--5751, Sep. 2023.

\bibitem{10206462}
F.~Liu, Y.~Xiong, K.~Wan, T.~X. Han, and G.~Caire, ``Deterministic-random
  tradeoff of integrated sensing and communications in {Gaussian} channels: A
  rate-distortion perspective,'' in \emph{Proc. 2023 IEEE Inte. Symp. Inf.
  Theory (ISIT)}, Taipei, Taiwan, Jun. 2023, pp. 2326--2331.

\bibitem{10596930}
S.~Lu, F.~Liu, F.~Dong, Y.~Xiong, J.~Xu, Y.-F. Liu, and S.~Jin, ``Random {ISAC}
  signals deserve dedicated precoding,'' \emph{IEEE Trans. Signal Process.},
  vol.~72, pp. 3453--3469, Jul. 2024.

\bibitem{10645253}
S.~Lu, F.~Liu, F.~Dong, Y.~Xiong, and K.~Guan, ``Optimal precoding toward
  random {ISAC} signals,'' \emph{IEEE Trans. Veh. Technol.}, vol.~73, no.~12,
  pp. 19\,884--19\,889, Dec. 2024.

\bibitem{10977963}
L.~Xie, F.~Liu, J.~Luo, and S.~Song, ``Sensing mutual information with random
  signals in {Gaussian} channels,'' \emph{IEEE Trans. Commun.}, vol.~73,
  no.~10, pp. 9437--9452, Oct. 2025.

\bibitem{11087656}
F.~Liu, Y.~Zhang, Y.~Xiong, S.~Li, W.~Yuan, F.~Gao, S.~Jin, and G.~Caire,
  ``{CP-OFDM} achieves the lowest average ranging sidelobe under {QAM/PSK}
  constellations,'' \emph{IEEE Trans. Inf. Theory}, vol.~71, no.~9, pp.
  6950--6967, Sep. 2025.

\bibitem{11391525}
X.~Song, X.~Yu, J.~Xu, and D.~W.~K. Ng, ``{CRB-rate} tradeoff for bistatic
  {ISAC} with {Gaussian} information and deterministic sensing signals,''
  \emph{IEEE Trans. Wireless Commun.}, vol.~25, pp. 11\,768--11\,782, Feb.
  2026.

\bibitem{xie2025bistatic}
\BIBentryALTinterwordspacing
L.~Xie, F.~Liu, S.~Song, and S.~Jin, ``Bistatic target detection by exploiting
  both deterministic pilots and unknown random data payloads,'' Aug. 2025.
  [Online]. Available: \url{https://arxiv.org/pdf/2508.18728}
\BIBentrySTDinterwordspacing

\bibitem{10380513}
G.~Cheng, Y.~Fang, J.~Xu, and D.~W.~K. Ng, ``Optimal coordinated transmit
  beamforming for networked integrated sensing and communications,'' \emph{IEEE
  Trans. Wireless Commun.}, vol.~23, no.~8, pp. 8200--8214, Aug. 2024.

\bibitem{sherman1950adjustment}
J.~Sherman and W.~J. Morrison, ``Adjustment of an inverse matrix corresponding
  to a change in one element of a given matrix,'' \emph{Ann. Math. Statist.},
  vol.~21, no.~1, pp. 124--127, Mar. 1950.

\bibitem{cvx}
\BIBentryALTinterwordspacing
M.~Grant and S.~Boyd, ``{CVX}: Matlab software for disciplined convex
  programming, version 2.1,'' Mar. 2014. [Online]. Available:
  \url{http://cvxr.com/cvx}
\BIBentrySTDinterwordspacing

\bibitem{9124713}
X.~Liu, T.~Huang, N.~Shlezinger, Y.~Liu, J.~Zhou, and Y.~C. Eldar, ``Joint
  transmit beamforming for multiuser {MIMO} communications and {MIMO} radar,''
  \emph{IEEE Trans. Signal Process.}, vol.~68, pp. 3929--3944, Jun. 2020.

\bibitem{9652071}
F.~Liu, Y.-F. Liu, A.~Li, C.~Masouros, and Y.~C. Eldar, ``{Cram\'er-Rao} bound
  optimization for joint radar-communication beamforming,'' \emph{IEEE Trans.
  Signal Process.}, vol.~70, pp. 240--253, Jan. 2022.

\bibitem{5447068}
Z.-Q. Luo, W.-K. Ma, A.~M.-C. So, Y.~Ye, and S.~Zhang, ``Semidefinite
  relaxation of quadratic optimization problems,'' \emph{IEEE Signal Process.
  Mag.}, vol.~27, no.~3, pp. 20--34, May 2010.

\bibitem{pataki1998rank}
G.~Pataki, ``On the rank of extreme matrices in semidefinite programs and the
  multiplicity of optimal eigenvalues,'' \emph{Math. Oper. Res.}, vol.~23,
  no.~2, pp. 339--358, May 1998.

\bibitem{10086626}
H.~Hua, J.~Xu, and T.~X. Han, ``Optimal transmit beamforming for integrated
  sensing and communication,'' \emph{IEEE Trans. Veh. Technol.}, vol.~72,
  no.~8, pp. 10\,588--10\,603, Aug. 2023.

\bibitem{4305444}
E.~Karipidis, N.~D. Sidiropoulos, and Z.-Q. Luo, ``Far-field multicast
  beamforming for uniform linear antenna arrays,'' \emph{IEEE Trans. Signal
  Process.}, vol.~55, no.~10, pp. 4916--4927, Oct. 2007.

\end{thebibliography}

\end{document}